\documentclass[technote, 12pt, letterpaper, draftclsnofoot, onecolumn]{IEEEtran}	

\usepackage[utf8]{inputenc}

\usepackage{amsmath}
\usepackage{amssymb}
\usepackage{subfigure}
\usepackage{graphicx}
\usepackage{graphics}
\usepackage{color}
\usepackage{psfrag}
\usepackage{cite}
\usepackage{balance}
\usepackage{algorithm}
\usepackage{accents}
\usepackage{amsthm}
\usepackage{bm}
\usepackage{url}
\usepackage{algorithmic}
\usepackage[english]{babel}
\usepackage{multirow}
\usepackage{enumerate}
\usepackage{cases}
\usepackage{stfloats}
\usepackage{dsfont}
\usepackage{soul}
\usepackage{amsfonts}
\usepackage{fancyhdr}
\usepackage{hhline}
\usepackage{array}
\usepackage{booktabs}
\usepackage{framed}
\usepackage{float}
\usepackage{soul}

\newtheorem{theorem}{Theorem}

\newtheorem{definition}{Definition}

\newtheorem{lemma}{Lemma}
\newtheorem{corollary}{Corollary}

\newtheorem{proposition}{Proposition}

\newtheorem{assumption}{Assumption}

\newtheorem{remark}{\bf Remark}
\def\phi{\varphi}

\def\l{\left}
\def\r{\right}

\def\bg{{\mathbf{g}}}
\def\bh{{\mathbf{h}}}

\def\br{{\mathbf{r}}}
\def\bs{{\mathbf{s}}}

\def\bu{{\mathbf{u}}}

\def\bw{{\mathbf{w}}}
\def\bx{{\mathbf{x}}}

\def\b0{{\mathbf{0}}}

\def\sT{{\mathsf{T}}}
\def\sH{{\mathsf{H}}}
\def\cmp{{\text{cmp}}}
\def\cmm{{\text{cmm}}}

\def\lo{{\text{lo}}}
\def\gl{{\text{gl}}}

\def\tP{{\widetilde{P}}}

\def\tbh{{\widetilde{\mathbf{h}}}}
\def\tbg{{\widetilde{\mathbf{g}}}}

\def\mC{{\mathbb{C}}}

\def\mE{{\mathbb{E}}}

\def\mI{{\mathbb{I}}}

\def\mR{{\mathbb{R}}}

\def\cD{\mathcal{D}}

\def\cK{\mathcal{K}}

\def\cM{\mathcal{M}}

\newcommand{\var}{\mathsf{Var}}

\newcommand{\tr}{\mathsf{tr}}

\begin{document}

\title{\Large Wirelessly Powered Federated Edge Learning:\\Optimal Tradeoffs Between Convergence and Power Transfer}
\author{
Qunsong Zeng, Yuqing Du, and Kaibin Huang
\thanks{Q. Zeng, Y. Du, and K. Huang are with The University of Hong Kong, Hong Kong.  Contact: K. Huang (huangkb@eee.hku.hk).}
}

\maketitle
\vspace{-2mm}
\begin{abstract}
\emph{Federated edge learning} (FEEL) is a widely adopted framework for training an \emph{artificial intelligence} (AI) model distributively at edge devices to leverage their data  while preserving their data privacy. The execution of a power-hungry  learning task at energy-constrained devices is a key challenge confronting the implementation of FEEL. To tackle the challenge, we propose the solution of powering devices using \emph{wireless power transfer} (WPT). To derive guidelines on deploying the resultant  \emph{wirelessly powered FEEL} (WP-FEEL) system, this work aims at the derivation of the  tradeoff between the model convergence and the settings of power sources in two scenarios: 1) the transmission power and density of power-beacons (dedicated charging stations) if they are deployed, or otherwise 2) the transmission power of a server (access-point). The development of the proposed analytical framework relates the accuracy of distributed stochastic gradient estimation to the WPT settings, the randomness in both communication and WPT links, and devices' computation capacities. Furthermore, the local-computation at devices (i.e., mini-batch size and processor clock frequency) is optimized  to efficiently use the harvested energy for gradient estimation. The resultant learning-WPT tradeoffs reveal the simple scaling laws of the model-convergence rate with respect  to the transferred energy as well as the devices' computational energy efficiencies. The results provide useful guidelines on WPT provisioning to provide a guaranteer on learning performance. They are corroborated  by experimental results using a real dataset. 

\end{abstract}

\section{Introduction}
Recent years have seen a growing trend of deploying machine learning algorithms at the wireless network edge to distill \emph{artificial intelligence} (AI) from the abundant data at edge devices (e.g., sensors and smart phones), giving rise to an area called \emph{edge learning} \cite{xu2019edgeAI,gxzhu_2018_edge_learning}. Among others, \emph{federated edge learning} (FEEL) is perhaps the most widely adopted  framework for its feature of preserving data privacy  \cite{niyato2020feel,wang2019adaptive,deniz2019federated_edge_learning}. Specifically, instead of uploading data from devices, the framework involves a server  distributing  a learning task over devices based on distributed implementation of \emph{stochastic gradient descent} (SGD). One challenge confronting federated learning in practice is that executing  a complex task (e.g., training of a large-scale \emph{convolutional neural network} (CNN)) at edge devices drains their batteries. To tackle this challenge, we propose the  solution of deploying \emph{wireless power transfer} (WPT) to deliver to devices the energy they need for computation and communication. To understand the performance of the resultant \emph{wirelessly powered FEEL} (WP-FEEL) system, this work aims at quantifying  the optimal tradeoffs between model convergence and settings of power sources, which can be power-beacons (charging stations) or the server, when devices optimally allocate harvested energy for computation and communication to accelerate convergence. The derived tradeoffs, termed the optimal \emph{learning-WPT tradeoffs}, yield useful insights into system design and deployment. 

The current research on  implementing FEEL in wireless networks  can be separated into two main thrusts focusing on tackling two different challenges. One challenge is the communication bottleneck arising from the wireless uploading of high-dimensional model updates (either local models or local stochastic gradients) from many devices. Attempts to overcome the bottleneck have led to the design of a new class of communication techniques for efficient FEEL including over-the-air updates aggregation  \cite{deniz2019federated_edge_learning,gxzhu2018FEEL,yang2019aircomp}, resource management \cite{chen2019joint,jinke2021,dingzhu2020}, adaptive uploading frequency control \cite{wang2019adaptive}, device scheduling \cite{yang2020}, and quantization \cite{yq2020quantization}. The other challenge is to execute energy consuming tasks at edge devices as mentioned earlier. This issue has been addressed in a series of works on designing techniques for jointly managing computation and communication resources \cite{sun2020energy,yang2019energy,mo2020energyefficient,qs2020cpu-gpu} under the criterion of minimizing the total devices' energy consumption during the learning process. Addressing the same issue, we propose an alternative and direct approach of powering devices using WPT.  Though there exist rich convergence analysis in the prior work, the  tradeoffs between energy consumption of devices  and convergence have not yet been crystallised. The derivation of the desired learning-WPT tradeoff is even more complex due to new issues arising from WPT especially the following two. First, the unreliabilities of  both communication and WPT links jointly affect the number of active devices. Second,  each device needs manage harvested energy for both communication and computation. Their coupling results in the channel dependence of  local computation (i.e., mini-batch size and processor frequency) and hence learning performance.

There exists a rich literature on the application of microwave based WPT to power different types of wireless networks ranging from communication networks to sensor networks to those supporting mobile-edge computing (see recent surveys in \cite{clerckx2019fundamental,clerckx2021wireless}). There exist three main topologies of wirelessly powered networks \cite{kaibin2015magazine}: 1) the integration of WPT with downlink transmission, called \emph{simultaneous wireless information and power transfer} (SWIPT) (see e.g., \cite{zhang2013mimo}), 2) downlink WPT to power uplink transmission (see e.g., \cite{ju2014WPT}), and 3) separated WPT served by power-beacons and radio access by base stations (see e.g., \cite{kaibin2014powerbeacon}). As the communication bottleneck  of a FEEL system lies in the uplink,  the last two topologies are relevant,  both of which are  considered in this paper. Despite building on the existing network topologies, WP-FEEL systems differ from their conventional wirelessly powered communication systems in several aspects. First, the performance of the former is measured using learning related metrics (i.e., convergence rate or test accuracy) and that of the the latter is measured by communication related metrics such as   throughput (see e.g., \cite{ju2014WPT}), communication energy efficiency (see e.g., \cite{energy-efficiency}), and  rate-harvested-energy tradeoff (see e.g., \cite{rate-energy-tradeoff}). Second, the devices in a WP-FEEL system are workers cooperating  in training a global model while those in a communication system are subscribers competing for power transfer and the use of radio resources. Third, computing power consumption is either neglected or abstracted as a constant for conventional systems focusing on communication (see e.g., \cite{deniz2019federated_edge_learning,sun2020energy,chen2019joint}). In contrast, such consumption is at least comparable with its communication counterpart  in a WP-FEEL system performing a computation intensive task. Thus, an elaborate model of the former is adopted in this work so that the analysis can be  of practical relevance. 

The above distinctions between  WP-FEEL systems and their conventional counterparts give rise to new challenges in designing and analyzing the former. To tackle the  challenges, the main contribution of this work is the development of a novel analytical framework for quantifying  the optimal learning-WPT tradeoff of a WP-FEEL system. As a by-product, a scheme for the optimal control of local computation at devices is designed.  The framework is first developed for the scenario where dense power-beacons are deployed to provide reliable WPT without fading, referred to as the \emph{beacon-WPT} \cite{kaibin2014powerbeacon}. The key components of the framework and relevant findings are described as follows. 

\begin{enumerate}
\item {\bf Distributed Gradient Estimation}: \emph{Global and local gradient deviations} are respectively defined as the expected deviation of a local gradient estimate and the global estimate from the ground truth computed using the global/local datasets. In existing convergence analysis, they are usually studied under the following assumptions:  1) i.i.d. data distributions at devices, 2) uniform mini-batch sizes,  and 3) a fixed number of active devices (see e.g., \cite{zhu2020onebit,yu2019parallel,basu2019qsparse,koloskova2019decentralized,ijcai2018}). While the first assumption lacks generality, the last two do not hold for  the WP-FEEL system featuring  random harvested energy, the mentioned channel-dependent heterogeneous   computation capacities, and a random number of active  devices. To address the issue, we define a generalized  system of global and local gradient deviations by relaxing the assumptions. By analyzing these measures, the convergence rate is related to the distribution of the set of active devices as well as the derived probability of a \emph{computation-outage event}, which occurs when  a device fails to harvest sufficient energy to support both communication and computation and hence becomes inactive. 

\item {\bf Local-Computation Optimization}: Consider an active device and an arbitrary  round. After reserving sufficient transmission energy, the remaining harvested energy is used for local computation. Under the energy constraint, the mini-batch size and processor's computing speed are jointly optimized to minimize the local gradient deviation. They are shown to both increase \emph{sub-linearly} with the computation energy   and be inversely  proportional to the device's computation capacity. In addition, the optimal mini-batch size is also inversely proportional to the workload for local gradient computation. 

\item {\bf Optimal Learning-WPT Tradeoff}: The tradeoff is derived based on characterizing the effects of WPT on distributed gradient estimation and devices' computation-and-communication  capacities.  Define the \emph{spatial energy density} for beacon-WPT as the total energy transferred from beacons to a randomly located device per round, denoted as $\lambda_{\text{energy}}$. The difference between the convergence rate and  its  limit in the ideal case of using the global dataset is found to be inversely proportional to a sub-linear function of  spatial energy density, namely $O\left(\lambda_{\text{energy}}^{-\frac{1}{3}}\right)$. The  result provides some guidelines on power-beacon deployment (i.e., power and density) to provide a guaranteer on the learning performance. Moreover, the difference is shown to decay as  a weighted sum of sub-linear functions of individual computation energy efficiencies (i.e., required energy for processing a data sample). Each weight depends on the usefulness of a  local dataset and specifically is the local gradient deviation for a single sample. The result suggests the need of considering devices' computation energy efficiencies in WPT provisioning. 

\end{enumerate} 

The framework  is extended to the other scenario of \emph{server-WPT}, where the server   transfers power to devices over fading channels \cite{ju2014WPT}. In particular, the scaling laws described above remain the same except that the spatial-energy density is replaced with the energy beamed by the server to each device in a specific  round.

The remainder of the paper is organized as follows. Mathematical models are introduced in Section II. In Section III, distributed gradient estimation is analyzed to relate the convergence rate to the distribution and computation capacities of active devices. The optimal learning-WPT tradeoff for beacon-WPT is derived in Section IV and extended to the scenario of server-WPT  in Section V. Experimental results are presented  in Section VI, followed by concluding remarks in Section VII.

\section{Mathematical Models}\label{system model}
We consider a single-cell WP-FEEL system in a circular cell with the radius denoted as $R$. A server equipped with an array of $L$ antennas  coordinates FEEL over  $K$ single-antenna edge devices, represented  by the  index set $\mathcal{K}=\{1,\cdots,K\}$. Devices are assumed to have high mobility and their locations are uniformly distributed in the cell and  i.i.d. over rounds. The devices are powered by either beacon-WPT or server-WPT as illustrated  in Fig.~\ref{Fig: WPT systems}. In each round with a fixed duration $T$, each device first computes a local gradient and then transmits it to the server. Then each round is divided in two phases: local computation and gradient uploading (see Fig.~\ref{Fig: WPT systems}), which last  ${T^{\cmp}}$ and $T^{\cmm}$ seconds, respectively. The operations of devices are synchronized, resulting in the following time constraints for edge devices:
\begin{equation}\label{Eqn: latency constraints}
    0<t_k^{\cmp}\leq {T^{\cmp}}\quad\text{and}\quad 0<t_k^{\cmm}\leq T^{\cmm},~\forall k\in\mathcal{K},
\end{equation}
where $t_k^{\cmp}$ and $t_k^{\cmm}$ are the computation and transmission time at device $k$ in one round, respectively. Let $E_k$ denote the amount of energy harvested by  device $k$; its computation and communication energy consumptions are represented by  $E_k^{\cmp}$ and $E_k^{\cmm}$, respectively. The harvested energy is fixed for beacon-WPT and varies over rounds for server-WPT as elaborated in the sequel. They satisfy  the following energy constraint:\footnote{To be precise, the idling circuit energy consumption, denoted as a constant $\zeta$, exists even when there is no computation and transmission. In this case, the energy constraint is $E_k^{\cmp}+E_k^{\cmm}+\zeta\leq E_k$. We omit the constant $\zeta$ as it is negligible compared with $E_k^{\cmp}$ and $E_k^{\cmm}$.}
\begin{equation}\label{Eqn: energy constraint}
    E_k^{\cmp}+E_k^{\cmm}\leq E_k,~\forall k\in\cK. 
\end{equation}
The detailed system operations and relevant models are described as follows. 

\subsection{Two WPT Models}
\subsubsection{Beacon-WPT}
Consider the WP-FEEL system powered by beacon-WPT in  Fig.~\ref{Fig: WPT systems}(a). Given their low cost and complexity, dense power-beacons are deployed to power devices over short-range WPT links without fading. The beacons  are modelled as a homogeneous \emph{Poisson point process} (PPP), denoted as $\Psi=\{\bs\}$ with density $\lambda_{\text{pb}}$, where $\bs\in\mR^2$ represent the coordinate of a single beacon. Each device is equipped with an energy harvester comprising a rectifying antenna and a battery \cite{kaibin2015magazine}. Moreover, WPT is over a dedicated frequency  outside the communication band. These allow the device to continuously harvest energy throughout the learning process [see Fig.~\ref{Fig: WPT systems}(a)]. Let the coordinates of device $k$ in the $i$-th round be denoted by $\br_k^{(i)}$ and thus the communication range  $r_k^{(i)} = |\br_k^{(i)}|$. Adopting a short-range propagation model \cite{baccelli2006MPTmodel},  the instantaneous power received at the device $k$ in round $i$ is given as
\begin{equation}
    P_k^{(i)}=\rho \bar{P}\sum_{\bs\in\Psi}\l(\max\{|\br_k^{(i)}-\bs|,\nu\}\r)^{-\beta},~\forall k\in\mathcal{K},
\end{equation}
where $\nu\geq1$ is a given constant avoiding singularity, $\beta>2$ is the path-loss exponent, $\bar{P}$ is the transmission power of power-beacons, and $\rho$ represents the product of energy-conversion efficiency and energy-beamforming gain. As the power-beacons are dense and homogeneously distributed in the cell,  the amount of harvested energy at each device in one round can be approximated as \cite{kaibin2014powerbeacon}
\begin{equation}\label{Eqn: power-beacon energy}
    E_k^{(i)}\approx\bar{E}=\frac{\pi\beta\rho\bar{P}\lambda_{\text{pb}}T}{(\beta-2)\nu^{\beta-2}},~\forall k\in\mathcal{K}.
\end{equation}
\begin{definition}{(Spatial-Energy Density)} \emph{The spatial-energy density is defined as  $\lambda_{\text{energy}} \triangleq \bar{P}\lambda_{\text{pb}}T$, which is proportional to the beacon density and transmission power. It can be interpreted as the amount of energy delivered by the power-beacon network to an arbitrarily located device in a single round.} 
\end{definition}

\begin{figure}[t!]
    \centering
    \subfigure[Beacon-WPT System and Operations]{\label{Fig: beacon-WPT}
    \includegraphics[width=0.98\textwidth]{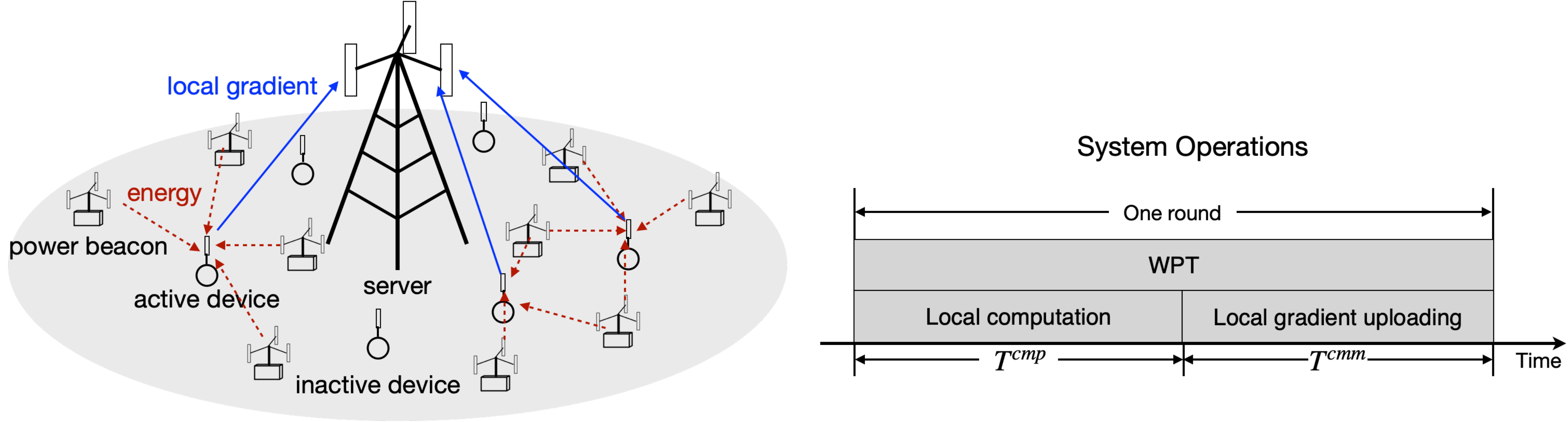}}
    \subfigure[Server-WPT System and Operations]{\label{Fig: server-WPT}
    \includegraphics[width=0.98\textwidth]{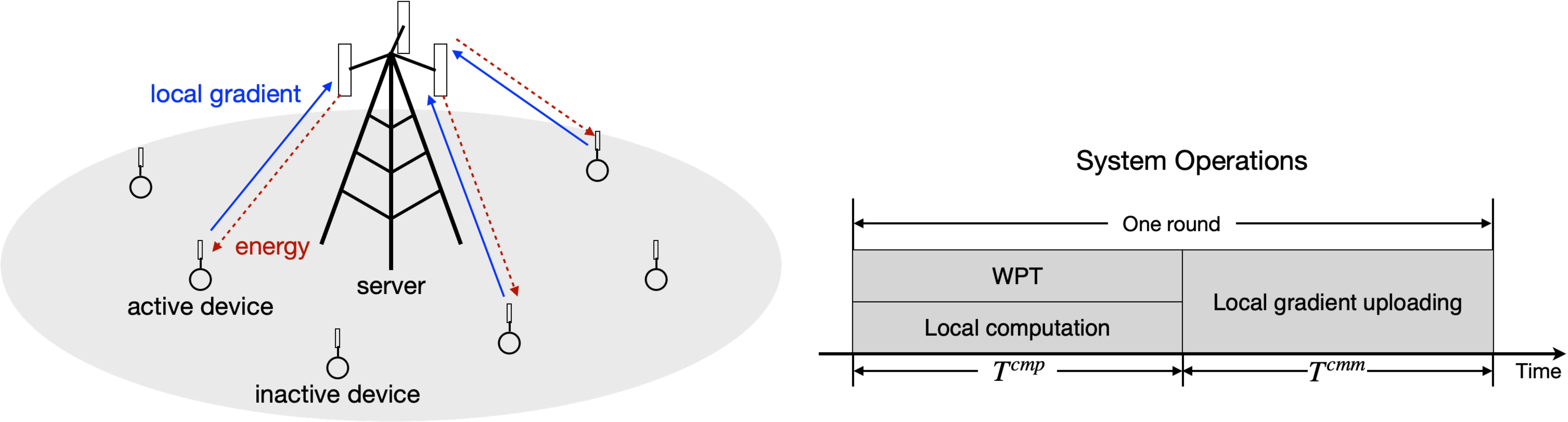}}
    \caption{WP-FEEL systems and operations: (a) beacon-WPT and (b) server-WPT.}
    \label{Fig: WPT systems}
\end{figure}

\subsubsection{Server-WPT}
In the absence of power-beacons, devices can be also powered by the server over long-range WPT links with fading as illustrated in  Fig.~\ref{Fig: WPT systems}(b). The server is assumed to be half-duplex and thus can perform WPT only during the local-computation phase in each round when its array is not used for communication [see Fig.~\ref{Fig: WPT systems}(b)].  Let the isotropic complex Gaussian vector  $\tbh_k^{(i)}\in\mC^{L\times 1}$ represent the Rayleigh fading channel of the WPT  link from the server to device $k$. Moreover,  let $\bu_k^{(i)}\in\mC^{L\times 1}$ with $(\bu_k^{(i)})^{\sH}\bu_k^{(i)}=1$ denote the energy-beamforming  vector, and $\tP_k^{(i)}$ the transfer power allocated to device $k$. With energy beamforming, the amount of harvested energy by device $k$ in each round is given as 
\begin{equation}
    E_k^{(i)}=\rho \big(r_k^{(i)}\big)^{-\alpha}\|\tbh_k^{(i)}\|^2\tP_k^{(i)} T^{\cmp},~\forall k\in\mathcal{K}. 
\end{equation}
Last, the WPT channels are assumed to be i.i.d. over rounds and furthermore independent of the uplink channels since they are in  different frequency bands. 

\subsection{Federated Learning Model}
A standard federated learning framework is considered as follows (see e.g., \cite{mcmahan2017federatedlearning}). A global model, represented by the parametric vector  $\mathbf{w}\in\mR^q$ with $q$ denoting  the model size, is trained collaboratively across the edge devices by leveraging local labelled datasets. For device $k$, let $\mathcal{D}_k = \{(\mathbf{x}_j,y_j)\}$ denote the local dataset where $\mathbf{x}_j$ and $y_j$ represent the raw data and label of the $j$-th sample. The \emph{local loss function} is defined as
\begin{equation}\label{Eqn: local loss function}
    F_k(\mathbf{w})=\frac{1}{|\mathcal{D}_k|}\sum_{(\mathbf{x}_j,y_j)\in\mathcal{D}_k}\ell(\mathbf{w};(\mathbf{x}_j,y_j)),
\end{equation}
where $\ell(\mathbf{w};(\mathbf{x}_j,y_j))$ is the sample-wise loss function quantifying the prediction error of the model $\mathbf{w}$ on the training sample $\mathbf{x}_j$ with reference to its true label $y_j$. For convenience, we denote $\ell(\mathbf{w};(\mathbf{x}_j,y_j))$ as $\ell_j(\mathbf{w})$ and assume uniform sizes for local datasets: $|\mathcal{D}_k|=D,\forall k\in\mathcal{K}$. Then the \emph{global loss function} on all the distributed datasets can be written as
\begin{equation}\label{eqn: global loss function}
    F(\mathbf{w})=\frac{\sum_{k=1}^K\sum_{j\in\mathcal{D}_k}\ell_j(\mathbf{w})}{\sum_{k=1}^K|\mathcal{D}_k|}=\frac{1}{K}\sum_{k=1}^KF_k(\bw).
\end{equation}
Its gradient $\nabla F(\bw^{(i)})$ is referred to as the \emph{ground-truth gradient}. The learning process is to minimize the global loss function $F(\bw)$. To this end, each round aims at estimating $\nabla F(\bw^{(i)})$ distributively to facilitate SGD. 

We adopt the existing gradient-averaging implementation of FEEL with the key operations illustrated in Fig.~\ref{Fig: system model} and described as follows (see e.g., \cite{jinke2021}). In each round, say the $i$-th round, the server broadcasts the current model $\bw^{(i)}$ to all edge devices. Due to channel fading, only a subset of  devices, denoted as a subset $\cM^{(i)}\subseteq\cK$ with size $M^{(i)}=|\cM^{(i)}|$, can participate in learning in this specific round. Each device in $\cM^{(i)}$ computes a local estimate of the gradient of its local loss function by randomly sampling  its local dataset $\mathcal{D}_k$. We denote the sampled mini-batch local dataset as $\mathcal{B}_k^{(i)}$ whose size is denoted by $b_k^{(i)}= |\mathcal{B}_k^{(i)}|$. The local gradient at device $k$ in the $i$-th round is estimated using the mini-batch as
\begin{equation}
    \bg_k^{(i)}=\frac{1}{b_k^{(i)}}\sum_{(\bx_j,y_j)\in\mathcal{B}_k^{(i)}}\nabla \ell_j(\bw^{(i)}).\label{eqn: stochastic gradient}
\end{equation}
Upon completion, the local gradient estimates are sent by active devices to the server for aggregation. Upon receiving them, the global gradient is calculated as
\begin{equation}\label{eqn: gradients aggregation}
    \bg^{(i)}=\begin{cases}\frac{1}{M^{(i)}}\sum\limits_{k\in\cM^{(i)}}\bg_k^{(i)},&M^{(i)}>0\\\mathbf{0},&M^{(i)}=0\end{cases}.
\end{equation}
Subsequently, the global model is then updated using SGD as
\begin{equation}\label{eqn: update rule}
    \bw^{(i+1)}=\bw^{(i)}-\eta\bg^{(i)},
\end{equation}
where $\eta$ is the given  learning rate. The process iterates until the model converges. In the process, the accuracy of distributed gradient estimation can be measured using the following metrics. 

\begin{figure}[t!]
    \centering
    \includegraphics[width=0.98\textwidth]{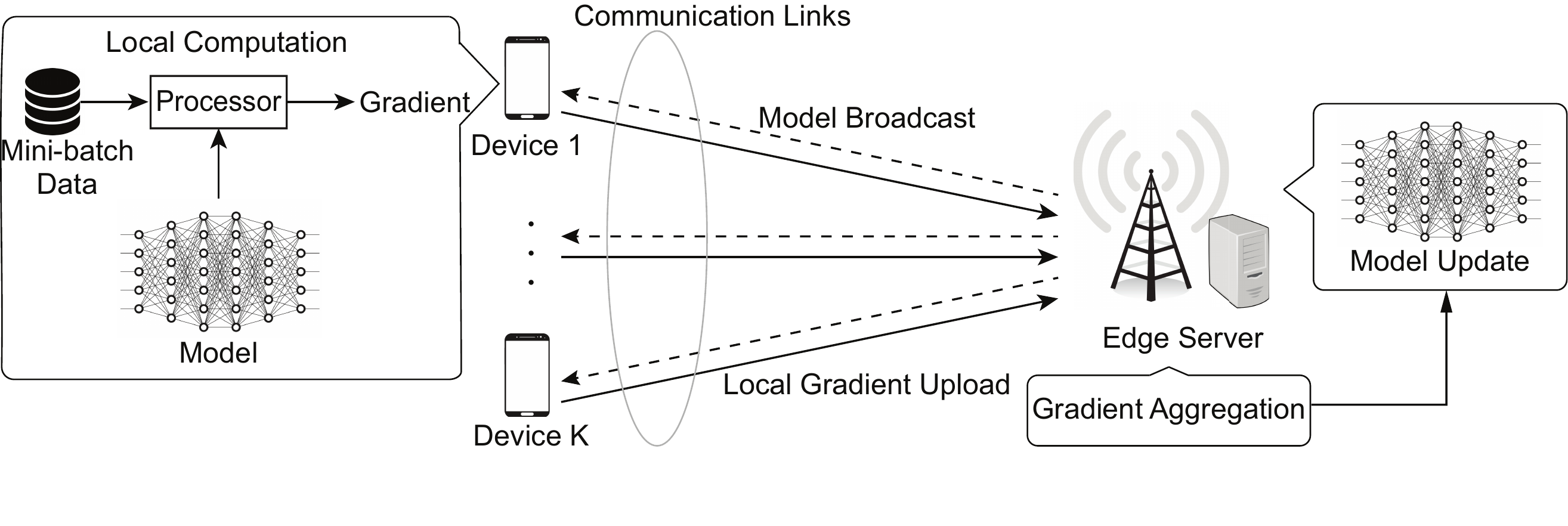}
    \vspace{-10mm}
    \caption{FEEL Operations.}
    \label{Fig: system model}
\end{figure}

\begin{definition}\label{definition: gradient deviation}
(Local and Global Gradient Deviations). \emph{In the $i$-th round, the local gradient deviation at device $k$, denoted as     $G_{\lo,k}^{(i)}$,  refers to the mean-square error between the local gradient estimate and its ground-truth:
\begin{equation}\label{Eqn: local gradient deviation definition}
    G_{\lo,k}^{(i)}=\mE\l[\|\bg_k^{(i)}-\nabla F_k(\bw^{(i)})\|^2\r].
\end{equation}
The global gradient deviation refers to the expected deviation between the aggregated  local gradient estimates at  the server and the ground truth: 
\begin{equation}\label{Eqn: gradient deviation definition}
    G_{\gl}^{(i)}=\mE_{\cM^{(i)}}\l\{\mE\l[\|\bg^{(i)}-\nabla F(\bw^{(i)})\|^2\r]\r\},
\end{equation}
where the outer expectation is taken over rounds and the distributions of $M^{(i)}$ and $b_k^{(i)}$.
}
\end{definition}

\subsection{Local-Computation Model}
The computation-energy consumption depends on two variables: 1) the mini-batch size  and 2) the processor's clock frequency. Adopting a standard model in computer engineering \cite{zhang2018flops}, we define the per-sample workload $W$ for local-gradient estimation as the number of \emph{floating point operations} (FLOPs)  needed for processing each data sample. This gives the workload at device $k$ in the $i$-th round as $W_{k,{\sf total}}^{(i)}=b_k^{(i)}\times W$.  Let $f_{{\sf clk},k}^{(i)}$ [in cycle/s] represent  the clock frequency of the processor (e.g., CPU or GPU) at device $k$ in round $i$.  As a result, the computing speed of the processor, measured in FLOPs per second, can be defined as $f_k^{(i)}=f_{{\sf clk},k}^{(i)}\times {N}_k^{\sf FLOP}$ with ${N}_k^{\sf FLOP}$ denoting the number of FLOPs it can process per cycle.  Given the workload and computing speed, the local computation time at device $k$, denoted as ${t_k^{\cmp}}^{(i)}$, is given by
\begin{equation}\label{Eqn: local computation time}
    {t_k^{\cmp}}^{(i)}=\frac{b_k^{(i)}W}{f_k^{(i)}},~\forall k\in\mathcal{K}.
\end{equation} 
For a CMOS circuit, the power consumption of a processor can be modelled as a function of clock frequency: $P=\Psi f_{\sf clk}^3$, where $\Psi$ [in $\text{Watt}/(\text{cycle/s})^3$] is a constant depending  on the chip architecture \cite{liu2012dvfs}. Based on this model, the power consumption of the processor at device $k$ can be written as 
\begin{equation}\label{Eqn: local computation power}
    {P_k^{\cmp}}^{(i)}=\Psi_k\big(f_{{\sf clk},k}^{(i)}\big)^3=C_k\big(f_k^{(i)}\big)^3,~\forall k\in\mathcal{K},
\end{equation}
where the coefficient $C_k=\Psi_k/\l(N_k^{\sf FLOP}\r)^3$ characterizes the computation property of the processor. In particular, a smaller value indicates that the processor is capable to compute more workload given energy consumption per unit time, or consume less energy given the workload per unit time. Given  \eqref{Eqn: local computation time} and  \eqref{Eqn: local computation power}, the resultant energy consumption at device $k$ is given as
\begin{equation}\label{eqn: computation energy}
    {E_k^{\text{cmp}}}^{(i)}={P_k^{\text{cmp}}}^{(i)}\times {t_k^{\cmp}}^{(i)}=b_k^{(i)}C_kW\big(f_k^{(i)}\big)^2,~\forall k\in\mathcal{K}.
\end{equation}

\subsection{Transmission Model}
Without loss of generality, consider uploading by device $k$ in the $i$-th round. Each gradient coefficient  is compressed into $Q$ bits such that the effect of quantization on the learning performance is negligible.  Then the overhead of transmitting a $q$-dimensional vector  is $q\times Q$ bits. The uplink bandwidth is equally divided into $K$ narrow sub-bands of $B$ and allocated to  the devices for orthogonal transmission. Let the  complex Gaussian vector $\bh_{k}^{(i)}$ comprising i.i.d. $\mathcal{CN}(0, 1)$ coefficients  represent  the  Rayleigh fading channel of the considered devices. Channels of different devices are assumed independent of each other.  Given receive beamforming at the server, the transmission rate for device $k$ in round $i$ can be written as
\begin{equation}
    S_k^{(i)}=B\log_2\left(1+\frac{\|\mathbf{h}_k^{(i)}\|^2{P^{\cmm}_k}^{(i)}}{\big(r_k^{(i)}\big)^{\alpha}BN_0}\right),~\forall k\in\mathcal{K}.
\end{equation}
where ${P_k^{\cmm}}^{(i)}$ represents the transmission power, $N_0$ the power spectrum density of the additive white Gaussian noise,  $r_k^{(i)}$ the propagation distance,  and $\alpha$ the path-loss exponent.  The transmission rate is required to support uploading of $q\times Q$ bits in a single round. This places the following  constraint on the transmission power:
\begin{equation}
    {P_k^{\text{cmm}}}^{(i)}=\frac{\big(r_k^{(i)}\big)^{\alpha}BN_0}{\|\mathbf{h}_k^{(i)}\|^2}\left(2^{\frac{qQ}{B{t_k^{\cmm}}^{(i)}}}-1\right),~\forall k\in\mathcal{K}.
\end{equation}
The resultant  transmission-energy consumption is 
\begin{equation}\label{eqn: transmission energy}
    {E_k^{\cmm}}^{(i)}={P_k^{\cmm}}^{(i)}\times{t_k^{\cmm}}^{(i)}=\frac{\big(r_k^{(i)}\big)^{\alpha}}{\|\mathbf{h}_k^{(i)}\|^2}\phi({t_k^{\cmm}}^{(i)}),~\forall k\in\mathcal{K},
\end{equation}
where the function $\varphi(t)\triangleq BN_0t\l(2^{\frac{qQ}{Bt}}-1\r)$.

\section{Convergence and Distributed Gradient Estimation}
Consider the WP-FEEL system  with  beacon-WPT. In this section, we aim at analyzing the relation between   convergence and  several key system variables influencing the accuracy of distributed gradient estimation, including the mini-batch sizes of devices, number of active devices, and computation-outage probability. Such results are useful for deriving the learning-WPT tradeoff in the next section. As direct analysis is difficult, a tractable approach is adopted using  global and local  gradient deviations as intermediate variables. 
\subsection{Convergence and Global Gradient Deviation}
To quantify the relation, we follow the literature to make several standard assumptions on the loss function and local estimated gradients as follows (see e.g., \cite{wang2019adaptive,bernstein2018signsgd,yu2019parallel,basu2019qsparse,koloskova2019decentralized,sattler2020noniid,zhu2018nonconvex,bottou2018}).
\begin{assumption}\label{assumption: smoothness}
(Smoothness). \emph{The loss function $F:\mathbb{R}^q\to\mathbb{R}$ is  \emph{$\mu$-smooth}. Specifically, for all $(\mathbf{u},\mathbf{v})\in\mathbb{R}^q\times\mR^q$, 
\begin{equation}
    F(\mathbf{u})\leq F(\mathbf{v})+\langle\nabla F(\mathbf{v}),\mathbf{u}-\mathbf{v}\rangle+\frac{\mu}{2}\|\mathbf{u}-\mathbf{v}\|^2,
\end{equation}
where $\nabla$ is the differential operator and $\langle\cdot,\cdot\rangle$ represents the inner product.
}
\end{assumption}
While  i.i.d. data distributions over  devices are commonly assumed in the literature for simplicity   (see e.g., \cite{zhu2020onebit}), we consider the general and more practical case of \emph{non-i.i.d.} data distribution as in \cite{sattler2020noniid}. For the case of i.i.d. distributions, local gradients are unbiased with respect to the global full-batch gradient:  $\mathbb{E}[\mathbf{g}_k]=\nabla F(\mathbf{w}),\forall k\in\mathcal{K}$, where the expectation is taken over the data distribution at device $k$. On the other hand, for the current case of non-i.i.d. distributions,  we make the following assumption on local gradient estimates \cite{sattler2020noniid}. 

\begin{assumption}\label{assumption: local gradients estimation}
(Local Gradients Estimation). \emph{The stochastic gradient estimates $\{\mathbf{g}_k\}$ defined in \eqref{eqn: stochastic gradient} are unbiased estimates of its local gradients $\{\nabla F_k(\mathbf{w})\}$ defined in \eqref{Eqn: local loss function} and computed using full local datasets, and independent of each other: 
\begin{equation}
    \mathbb{E}[\mathbf{g}_k]=\nabla F_k(\mathbf{w}),~\forall k\in\mathcal{K},
\end{equation}
where the expectation is taken over local data distribution $\cD_k$. 
}
\end{assumption}
It should be reiterated  that local  gradients are not equal to the global gradient and their relation is specified in \eqref{eqn: global loss function}.  Our analysis does not require the convexity assumption for the  loss function and only requires it to be lower bounded as formally stated below, which is the minimal assumption needed for ensuring convergence to a stationary point \cite{zhu2018nonconvex}.
\begin{assumption}\label{assumption: bounded loss function}
(Bounded Loss Function). \emph{For any parameter vector $\mathbf{w}$, the loss function $F(\bw)$ is  lower bounded by a given  scalar $F_*$.
}
\end{assumption}
The last assumption given below is also standard in the literature   (see e.g., \cite{yu2019parallel,basu2019qsparse,koloskova2019decentralized}).

\begin{assumption}\label{assumption: bounded gradient norm}
(Bounded Gradient Norm). \emph{The expected squared norm of stochastic gradients is uniformly bounded by a constant $\Phi$, that is, $\mE\l[\|\bg_k^{(i)}\|^2\r]\leq\Phi$, $\forall k\in\cK$ and $\forall i$.
}
\end{assumption}


We adopt a widely metric for measuring the convergence rate of FEEL with a non-convex loss function, namely the \emph{expected average  gradient norm} (over rounds) \cite{ijcai2018,bottou2018,bernstein2018signsgd}. Based on the above assumptions, we prove that expected  average  gradient norm can be bounded by the average global gradient deviation  as shown in the following proposition, thereby relating convergence to gradient estimation.  

\begin{proposition}\label{proposition: convergence rate} 
\emph{Given the learning rate satisfying $0<\eta\leq\frac{1}{\mu}$, the expected convergence rate of the FEEL algorithm can be upper-bounded by the average global gradient deviation as follows 
\begin{equation}\label{Eqn: convergence rate}
\begin{aligned}
    \mE\l[\frac{1}{N}\sum_{i=0}^{N-1}\|\nabla F(\bw^{(i)})\|^2\r]&\leq\frac{2\l[F(\bw^{(0)})-F_*\r]}{\eta N}+ \frac{1}{N}\sum_{i=0}^{N-1}G_{\gl}^{(i)}.\\
\end{aligned}
\end{equation}
where $G_{\gl}^{(i)}$ is  defined in \eqref{Eqn: gradient deviation definition}.
}
\end{proposition}
\begin{proof}
See Appendix \ref{proof: convergence rate}.
\end{proof}

\subsection{Computation-Outage Probability}
The computation-outage probability that affects the global gradient deviation is derived as follows. Without loss of generality, consider device $k$. 

\begin{definition}\label{definition: outage event}
(Computation-Outage Event). \emph{The event occurs at device $k$ in the $i$-th round when its harvested energy $\bar{E}$ is no larger than  the required  transmission energy $E_k^{\cmm}(r_k^{(i)},\bh_k^{(i)})$ given the propagation distance $r_k^{(i)}$ and fading channel $\bh_k^{(i)}$. As a result, there is zero energy for computation and  device $k$ is inactive in round $i$:  $k\notin \cM^{(i)}$.
}
\end{definition}
It is well known that the gain of the  Rayleigh fading channel,   $\|\bh_k^{(i)}\|^2$,  follows the $\chi^2$-distribution with the following   \emph{probability density function} (PDF): 
\begin{equation}\label{Eqn: channel pdf}
    f_{\|\bh_k^{(i)}\|^2}(h)=\frac{h^{L-1}e^{-h}}{\Gamma(L)},\quad h\geq 0,
\end{equation}
where $\Gamma(\cdot)$ is the Gamma function. On the other hand, since each  device is uniformly distributed in the cell, the transmission   distance of   device $k$  has  the following  PDF:
\begin{equation}\label{Eqn: location pdf}
    f_{r_k^{(i)}}(r)=\frac{2r}{R^2},\quad 0\leq r\leq R.
\end{equation}
Next, since transmission energy  is a monotone decreasing function of the transmission duration, to minimize energy consumption requires the use of the maximum transmission duration: ${t_k^{\cmm}}^{(i)} = T^{\cmm}$. Then the  required transmission energy  in round $i$ is  $E_k^{\cmm}(r_k^{(i)},\bh_k^{(i)})=\frac{(r_k^{(i)})^{\alpha}}{{\|\bh_k^{(i)}\|^2}}\varphi(T^{\cmm})$.  Using the above results, the computation-outage probability is derived as follows.

\begin{lemma}\label{lemma: outage probability}
(Computation-Outage Probability). \emph{The  probability is identical for all devices and all round and  given as
\begin{equation}\label{Eqn: outage probability}
    P_{\text{out}}=\frac{\gamma(L,\xi)-\xi^{-\frac{2}{\alpha}}\gamma(L+\frac{2}{\alpha},\xi)}{\Gamma(L)},
\end{equation}
where $\gamma(\cdot,\cdot)$ is the lower incomplete Gamma function, and the parameter $\xi$ is defined as
\begin{equation}\label{Eqn: xi}
    \xi\triangleq\frac{\varphi(T^{\cmm})}{R^{-\alpha}\bar{E}}=\frac{(\beta-2)\nu^{\beta-2}R^{\alpha}\varphi(T^{\cmm})}{\pi\beta\rho\bar{P}\lambda_{\text{pb}}T}.
\end{equation}
}
\end{lemma}
\begin{proof}
    See Appendix \ref{proof: outage probability}.
\end{proof}

The parameter $\xi$ defined in \eqref{Eqn: xi} is a key parameter influencing  $P_{\text{out}}$. The asymptotic scalings of $P_{\text{out}}$ with respect to $\xi$ are characterized in the following corollary of Lemma \ref{lemma: outage probability}.

\begin{corollary}\label{corollary: asymptotic outage probability}
 \emph{The computation-outage probability $P_{\text{out}}(\xi)$ is a monotone  increasing function of the parameter $\xi$. Asymptotically, $P_{\text{out}}(\xi)$ scales with respect to $\xi$  as follows:
\begin{equation}
    \lim_{\xi\to0}\frac{P_{\text{out}}(\xi)}{\xi^L}=\frac{2}{(\alpha L+2)\Gamma(L+1)}\quad\text{and}\quad
    \lim_{\xi\to\infty}\frac{1-P_{\text{out}}(\xi)}{\xi^{-\frac{2}{\alpha}}}=\frac{\Gamma(L+\frac{2}{\alpha})}{\Gamma(L)}.
\end{equation}
}
\end{corollary}
Based on the definition of $\xi$ in \eqref{Eqn: xi}, the above scaling laws suggest that the computation-outage probability can be reduced by 1) enhancing  the  harvested energy $\bar{E}$ via increasing the density and transmission power of power-beacons or 2) decreasing the cell size.

\begin{remark}
(Active and Idle Rounds). \emph{A round is idle with learning paused when all devices are in outage (i.e., $M=0$). The probability of idling round is $\Pr(M=0)=P_{\text{out}}^K$ and that  of active round is $ \Pr(M>0)=1-P_{\text{out}}^K$. }
\end{remark}

\subsection{Effects of System Variables on Convergence}
Given the result in Proposition \ref{proposition: convergence rate}, characterizing the effects requires only the analysis of the relation between the global gradient deviation and the system parameters.  To this end, consider an arbitrary round and the superscripts $i$ of variables, which specify the round index, are omitted in the remainder of the sub-section to simplify notation. An arbitrary active device, say device $k$, randomly draws a mini-batch of $b_k$ samples with the index set $\mathcal{B}_k=\{j_1,\cdots,j_{b_k}\}$. Then the  local gradient estimate can be written as  $\mathbf{g}_k=\frac{1}{b_k}\sum_{j\in\mathcal{B}_k}\nabla\ell_j(\mathbf{w})$. Its   distribution is specified  in   the following lemma.
\begin{lemma}\label{lemma: distribution of local gradient estimate}
(Distribution of Local Gradient Estimate  \emph{\cite{wu2020noisySGD}}). \emph{At active device $k$, the first two moments of the local gradient estimate  are given as $\mE[\bg_k]=\nabla F_k(\bw)$ and $\var[\bg_k]=\frac{\Omega_k}{b_k}$ with $\Omega_k$ being a constant defined as 
$\Omega_k\triangleq\frac{1}{D}\sum_{j\in\cD_k}\nabla\ell_j(\bw)\nabla\ell_j(\bw)^{\sT}-\nabla F_k(\bw)\nabla F_k(\bw)^{\sT}$.
}
\end{lemma}
It follows that the local gradient deviation at device $k$ can be written as $G_{\lo,k}=\frac{\sigma_k^2}{b_k}$  with its single-sample variance   $\sigma_k^2\triangleq \tr(\Omega_k)$.  Based on Assumptions \ref{assumption: local gradients estimation} and \ref{assumption: bounded loss function} and Lemma~\ref{lemma: distribution of local gradient estimate}, we obtain the following useful result. 
\begin{lemma}\label{proposition: global gradient deviation}
 \emph{The global gradient deviation in an arbitrary round can be bounded as 
\begin{equation}\label{Eqn: expected global GDve}
\begin{aligned}
    G_{\gl}\leq\frac{2}{K^2}\mE\l[\sum\nolimits_{k\in\cM}\frac{\sigma_k^2}{b_k}\r]+2\l\{\l(1-P_{\text{out}}^K\r)\l(\mE\l[\frac{1}{M}\Big|M>0\r]-\frac{1}{K}\r)+P_{\text{out}}^2\r\}\Phi,
\end{aligned}
\end{equation}
where $P_{\text{out}}$ is given in \eqref{Eqn: outage probability}. 
}
\end{lemma}
\begin{proof}
See Appendix \ref{proof: global gradient deviation}.
\end{proof}
Combining Proposition \ref{proposition: convergence rate} and Lemma \ref{proposition: global gradient deviation} gives the main result of this sub-section. 

\begin{proposition}\label{proposition: convergence} 
\emph{Given the learning rate satisfying $0<\eta\leq\frac{1}{\mu}$, the expected convergence rate of the FEEL algorithm satisfies 
\begin{equation}\label{Eqn: convergence}
\begin{aligned}
    \mE\l[\frac{1}{N}\sum_{i=0}^{N-1}\|\nabla F(\bw^{(i)})\|^2\r]&\leq\frac{2\l[F(\bw^{(0)})-F_*\r]}{\eta N}+ \frac{2}{K^2}\mE\l[\sum\nolimits_{k\in\cM}\frac{\sigma_k^2}{b_k}\r]\\
    & \quad +2\l\{\l(1-P_{\text{out}}^K\r)\l(\mE\l[\frac{1}{M}\Big|M>0\r]-\frac{1}{K}\r)+P_{\text{out}}^2\r\}\Phi.
\end{aligned}
\end{equation}
}
\end{proposition}

The above result relates convergence to several key system parameters, including the  mini-batch sizes, number of active devices and its distribution, and computation-outage probability.  In addition, one can observe that the upper bound in Proposition \ref{proposition: convergence} is identical for all rounds and thus simplify the subsequent analysis. 
\section{Optimal Learning-WPT Tradeoff  for Beacon-WPT}
In this section, we first give the optimized local computation policy, which determines how many samples the active devices can process in this round. Then, the global gradient deviation is derived by exploring the two factors that affect the global gradient deviation, respectively. The learning-energy tradeoff for beacon-WPT is characterized by the relation between convergence rate and spatial-energy density from the power-beacon network.

\subsection{Local Computation Optimization}\label{section: local computation optimization}
The local-computation variables at each device can be optimized to maximize the accuracy of local gradient estimation. On other hand, increasing the local batch size for model training reduces the local gradient deviation, but it increases local energy consumption according to \eqref{eqn: computation energy}. Moreover, under the time constraint, processing more samples requires boosting the computing speed, which also contributes to the energy growth. Thus it is useful  to control the two variables, sampled batch size and computing speed, under the criterion of minimum local gradient deviation. Considering device $k$ without loss of generality, since the local gradient deviation, $G_{\lo,k}$, given in \eqref{Eqn: local gradient deviation definition} is inversely proportional to the sampled batch size, $b_k$, the optimization problem can be formulated as 
\begin{equation*}\label{problem: local computation}{\bf (P1)}\quad
    \begin{aligned}
        \max_{\{b_k,f_k\}}\quad &b_k\\
        \text{s.t. }\quad&0<b_kC_kWf_k^2\leq E_k^{\cmp},\\
        &0<\frac{b_kW}{f_k}\leq {T^{\cmp}}.
    \end{aligned}
\end{equation*}
For tractability,  the batch size $b_k$ is relaxed to be continuous. This is reasonable as the batch size for model training is usually large  (e.g., thousands of images). 
With the relaxation, the optimal policy is derived in closed-form as shown below. 
\begin{lemma}\label{lemma: optimal computation policy}
(Optimal Local-Computation Policy). \emph{For optimal local gradient estimation  at each active device, the optimal sampled batch size and the computing speed should be set as follows: 
\begin{equation}
    b_k^\star=\frac{1}{W}\l(\frac{E_k^{\cmp}(T^{\cmp})^2}{C_k}\r)^{\frac{1}{3}}\quad\text{and}\quad f_k^\star=\l(\frac{E_k^{\cmp}}{C_k{T^{\cmp}}}\r)^{\frac{1}{3}},\quad \forall k\in\cM.
\end{equation}
}
\end{lemma}
The proof involves straightforward application of the \emph{Karush–Kuhn–Tucker} (KKT) conditions and is  omitted for brevity. 

Given the optimal policy, the expectation of the resultant local gradient estimate at an active device is derived as follows. 
First, the expectation  can be expressed in terms of  the computation energy budget $\{E_k^{\cmp}\}$ as
\begin{equation}\label{Eqn: gradient deviation vs E_k^cmp}
    \mE\l[G_{\lo,k}^{\star}\big|k\in\cM\r]=\frac{W}{(T^{\cmp})^{\frac{2}{3}}}\mE_{E_k^{\cmp}}\l[\frac{\sigma_k^2C_k^{\frac{1}{3}}}{(E_k^{\cmp})^{\frac{1}{3}}}\Bigg|k\in\cM\r].
\end{equation}
Since the  transferred energy is equal to the sum of computation and transmission energy, we can write $\{E_k^{\cmp}\}$ as the function of propagation distance and channel state  as $E_k^{\cmp}(r_k, \bh_k)=\bar{E}-\frac{r_k^{\alpha}}{\|\bh_k\|^2}\varphi(T^{\cmm})$. It follows from \eqref{Eqn: gradient deviation vs E_k^cmp} that 
\begin{equation}\label{Eqn: express expected local gradient}
    \Pr(k\in\cM)\mE\l[G_{\lo,k}^{\star}|k\in\cM\r]
    =\frac{W}{(T^{\cmp})^{\frac{2}{3}}}\iint_{\Theta}\frac{\sigma_k^2C_k^{\frac{1}{3}}}{\l(\bar{E}-E_k^{\cmm}(r,h)\r)^{\frac{1}{3}}}\frac{h^{L-1}e^{-h}}{\Gamma(L)}\frac{2r}{R^2}drdh,
\end{equation}
where the integral domain is defined as $\Theta=\{(r,h)\in\mR^2:\bar{E}>E_k^{\cmm}(r,h);~h\geq0;~0\leq r\leq R\}$. Based on \eqref{Eqn: express expected local gradient}, we obtain the following result.

\begin{lemma}\label{lemma: optimal local-gradient deviation}
(Optimal Local Gradient Deviation). \emph{Given optimal local computation in Lemma \ref{lemma: optimal computation policy}, the expectation of the  local gradient deviation at each active device can be bounded as
\begin{equation}\label{Eqn: upper bound for aggregated local gradient deviations}
\begin{aligned}
    \mE\l[G_{\lo, k}^{\star}|k\in\cM\r]&\leq \frac{2W\sigma_k^2C_k^{\frac{1}{3}}B(\frac{2}{3},\frac{2}{\alpha})}{\alpha (T^{\cmp})^{\frac{2}{3}}\bar{E}^{\frac{1}{3}}},
\end{aligned}
\end{equation}
where $B(\cdot,\cdot)$ is the Beta function, and $\bar{E}$ is the harvested energy at one device given in \eqref{Eqn: power-beacon energy}.
}
\end{lemma}
\begin{proof}
    See Appendix \ref{proof: upper bound for aggregated local gradient deviations}.
\end{proof}


\subsection{Optimal Learning-WPT Tradeoff}\label{section: optimal learning-WPT tradeoff}
So far we have analyzed two factors, local gradient deviation and computation-outage probability, which affect the global gradient deviation. To derive the desired learning-energy tradeoff with optimal local computation, we need to quantify the last factor,  $\mE\l[\frac{1}{M}\Big|M>0\r]$, with $M$ being the number of active devices (see Proposition \ref{proposition: global gradient deviation}). This factor represents the fact that more active devices help improve the accuracy of global gradient estimation. To analyze the factor, whether each device  participates  in a round can be represented by a Bernoulli random variable with the parameter $(1-P_{\text{out}})$, namely $\mI_{k\in\cM}\sim \text{Ber}(1-P_{\text{out}})$, $\forall k\in\cK$, where $\mI_{k\in\cM}$ denotes the indicator whose value is $1$ if $k\in\cM$, or $0$ otherwise. It follows that $M\sim \text{Binom}(K,1-P_{\text{out}})$, which  follows the Binomial distribution. The truncated version of the probability mass function with $M > 0$ is given as 
\begin{equation}
    \Pr(M=m)=\frac{1}{1-P_{\text{out}}^K}\binom{K}{m}(1-P_{\text{out}})^mP_{\text{out}}^{K-m},~m=1,\cdots,K.
\end{equation}
Using the distribution and a result from \cite{Stephan1945statistics},  the desired factor can be  obtained as shown in the following lemma.

\begin{lemma}\label{lemma: expected reciprocal}
(Expected Reciprocal \emph{\cite{Stephan1945statistics}}). \emph{The expected reciprocal of the number of active devices,  $M$, can be written in terms of the computation-outage probability as follows: 
\begin{equation}\label{Eqn: expected reciprocal}
    \mE\l[\frac{1}{M}\Big|M>0\r]=\frac{1}{1-P_{\text{out}}^K}\sum_{m=1}^K\frac{P_{\text{out}}^{m-1}-P_{\text{out}}^K}{K-m+1}.
\end{equation}
}
\end{lemma}


The global gradient deviation $G_{\gl}$ can be obtained from Lemma \ref{proposition: global gradient deviation} by substituting the results in Lemma \ref{lemma: optimal local-gradient deviation} and \ref{lemma: expected reciprocal}, where we emphasize that $\Pr(k\in\cM)=1-P_{\text{out}}$ and
\begin{equation}
    \mE\l[\sum_{k\in\cM}G_{\lo,k}^{\star}\r]=\sum_{k=1}^K\mE\l[\mI_{k\in\cM}G_{\lo,k}^{\star}\r]=\sum_{k=1}^K\Pr(k\in\cM)\mE\l[G_{\lo,k}^{\star}|k\in\cM\r].
\end{equation}
Then substituting $G_{\gl}$ into Proposition \ref{proposition: convergence rate} completes the proof.

Using the  preceding results, the optimal learning-energy tradeoff for the case of beacon-WPT can be readily derived as shown in the following theorem. 

\begin{theorem}\label{theorem: expected learning performance}
(Convergence with Beacon-WPT). \emph{Consider the case of beacon-WPT. Given the optimal local computation in Section \ref{section: local computation optimization},  the convergence rate of WP-FEEL is bounded by 
\begin{equation}\label{Eqn: convergence rate with PB-WPT}
    \mE\l[\frac{1}{N}\sum_{i=0}^{N-1}\|\nabla F(\bw^{(i)})\|^2\r]\leq\frac{2\l[F(\bw^{(0)})-F_*\r]}{\eta N}+\frac{\delta W\overline{C_{\sigma}}(1-P_{\text{out}})}{\rho^{\frac{1}{3}}K(T^{\cmp})^{\frac{2}{3}}\lambda_{\text{energy}}^{\frac{1}{3}}} + R_{\text{es}},
\end{equation}
where $\lambda_{\text{energy}}$ is the spatial-energy density, the parameter $\overline{C_{\sigma}}\triangleq \frac{1}{K}\sum_{k=1}^K\sigma_k^2C_k^{\frac{1}{3}}$, the constant  
$\delta\triangleq\frac{4}{\alpha}B(\frac{2}{3},\frac{2}{\alpha})\l(\frac{(\beta-2)\nu^{\beta-2}}{\pi\beta}\r)^{\frac{1}{3}}$,
the computation-outage probability $P_{\text{out}}$ specified  in Lemma \ref{lemma: outage probability}, and the residue term $R_{\text{es}}$ given as 
\begin{equation}\label{Eqn: residue}
R_{\text{es}} = 2\l(\sum_{m=2}^K\frac{P_{\text{out}}^{m-1}-P_{\text{out}}^K}{K-m+1}+P_{\text{out}}^2\r)\Phi. 
\end{equation}
}
\end{theorem}

The three terms on the right-hand side of \eqref{Eqn: convergence rate with PB-WPT} are explained as follows. The first term represents gradient descent using the ground-truth gradients. The second term, which arises from the global gradient deviation, reflects the effects of beacon-WPT and other system parameters on the convergence rate. The effects of individual parameters are discussed as follows. 

\begin{itemize}
    \item (WPT effect). Recall that the spatial-energy density, $\lambda_{\text{energy}} = \bar{P}\lambda_{\text{pb}}T$, refers to  the amount of energy transferred  by the power-beacon network to a unit area in a single round. Increasing the density leads to a linear growth of energy harvested by each device. As a result, active devices can estimate the local gradient with higher accuracies by  using  larger mini-batch sizes. This leads to the decrease of the global gradient deviation proportionally with $\lambda_{\text{energy}}^{-\frac{1}{3}}$. 
    
    \item (Local-computation properties). The computation parameters are grouped in $\frac{W\overline{C_{\sigma}}}{(T^{\cmp})^{\frac{2}{3}}}$. It can be interpreted as the reduction of global gradient deviation by processing a single sample at each device. The interpretation can be derived  mathematically by writing 
    \begin{equation}\label{Eqn: local computation properties}
        \frac{W\overline{C_{\sigma}}}{(T^{\cmp})^{\frac{2}{3}}}=\frac{1}{K}\sum_{k=1}^K\sigma_k^2\l[\frac{C_kW^3}{(T^{\cmp})^2}\r]^{\frac{1}{3}}\triangleq\frac{1}{K}\sum_{k=1}^K\sigma_k^2(e_k^{\cmp})^{\frac{1}{3}},
    \end{equation}
    where $e_k^{\cmp}$ represents the required energy for per-sample computation at device $k$. While the above quantity quantifies the per-sample gain of distributed gradient estimation, the number of samples each device is capable of processing depends on the spatial-energy density of beacon-WPT discussed earlier. 
    
    \item (Number of devices). Due to the averaging operation in \eqref{eqn: gradients aggregation}, increasing the number of devices, $K$, reduces the global gradient deviation following the scaling law of $O\l(\frac{1}{K}\r)$.
    
    \item (Probability of participation). The quantity $(1-P_{\text{out}})$ represents the probability of participating in learning by  each device. As adding an active device contributes its local gradient deviation to the global counterpart,  increasing $(1-P_{\text{out}})$ seems to enlarge the latter. On the contrary, aligned with intuition, the overall effect is to reduce the global gradient deviation if taking into account  the last term, $R_{\text{es}}$, which is also part of the deviation and decreases as $(1-P_{\text{out}})$ grows. 
\end{itemize}
The last term $R_{\text{es}}$ captures the loss of convergence rate caused by those devices in computation-outage. If $P_{\text{out}}$ is small, $R_{\text{es}}$ scales as $O(P_{\text{out}})$, confirming the effect mentioned above  that $R_{\text{es}}$ decreases as the probability of participation, $(1 - P_{\text{out}})$, grows. 

The mentioned effect of WPT on the convergence of WP-FEEL can be mathematically quantified in the following corollary of Theorem \ref{theorem: expected learning performance}. 

\begin{corollary}\label{corollary: scaling law}
(Effect of WPT). \emph{Consider the case of beacon-WPT.  As  the spatial-energy density $\lambda_{\text{energy}}$ grows, the convergence rate increases as follows: 
\begin{equation}
    \mE\l[\frac{1}{N}\sum_{i=0}^{N-1}\|\nabla F(\bw^{(i)})\|^2\r]\leq\frac{2\l[F(\bw^{(0)})-F_*\r]}{\eta N}+\Upsilon \lambda_{\text{energy}}^{-\frac{1}{3}}+O\l(\lambda_{\text{energy}}^{-L}\r), \quad \lambda_{\text{energy}} \rightarrow \infty,
\end{equation}
where the parameter $\Upsilon \triangleq \frac{\delta W\overline{C_{\sigma}}}{\rho^{\frac{1}{3}} K(T^{\cmp})^{\frac{2}{3}}}$, and the outage probability scales as $P_{\text{out}}=O(\lambda_{\text{energy}}^{-L})$.
}
\end{corollary}

\section{Extension to Server-WPT}
The preceding sections focus on beacon-WPT. The results therein can be extended to the case of server-WPT by  accounting for fading in  WPT links.  Given channel-state information,  the power allocation at  the server for transfer to difference devices  can be optimized in the sequel. 

\subsection{Optimal Learning-WPT Tradeoff}

In the scenario  of server-WPT with  fading in the WPT links, the convergence analysis is more tedious than  the beacon-WPT counterpart.  Specifically, the current analysis differs from its beacon-WPT counterpart in two factors: 1) computation-outage probability accounting for fading in a WPT link, and 2) harvested energy that is now random.  For tractability, assume equal power allocation WPT, namely that each energy beam has  fixed power of $P_0$ (Section \ref{section: extended version} for power control). First, the computation-outage probability, denoted as $P_{\text{out}}'$, is  derived as follows. 
The definition of a computation-outage event can be modified from that in Definition \ref{definition: outage event} by replacing $\bar{E}$ with  $E_k(r_k^{(i)},\|\tbh_k^{(i)}\|^2)$. Specifically, a computation-outage event occurs if 
\begin{equation}\label{Eqn: outage event in server-WPT}
\frac{\|\tbh_k^{(i)}\|^2\|\bh_k^{(i)}\|^2}{\big(r_k^{(i)}\big)^{2\alpha}}\leq\frac{\varphi(T^{\cmm})}{\rho P_0T}.
\end{equation}
To simplify notation, we define the product random variable $X=\|\tbh_k^{(i)}\|^2\|\bh_k^{(i)}\|^2$. Since both variables in the product  follow the $\chi^2$-distribution with $2L$ degrees-of-freedom, $X$ has the following distribution:
\begin{equation}\label{Eqn: distribution of X}
    f_X(x)=\int_0^{\infty}f_{\|\tbh_k\|^2}(h)f_{\|\bh_k\|^2}(x/h)\frac{1}{h}dh=\frac{2x^{L-1}\mathsf{K}_0(2\sqrt{x})}{\Gamma(L)^2},
\end{equation}
where $\mathsf{K}_0(\cdot)$ is the modified Bessel function of the second kind. Combining \eqref{Eqn: outage event in server-WPT} and \eqref{Eqn: distribution of X} gives 
\begin{equation}
    P_{\text{out}}'=\int_0^{\tau}\frac{2x^{L-1}\mathsf{K}_0(2\sqrt{x})\l(1-(x/\tau)^{\frac{1}{\alpha}}\r)}{\Gamma(L)^2}dx,
\end{equation}
where the parameter $\tau$ is defined as
    $\tau\triangleq\frac{R^{2\alpha}\varphi(T^{\cmm})}{\rho P_0T^{\cmp}}$.
As the exact expression of  $P_{\text{out}}'$ has no closed form, we  derive an upper bound  for tractability.
\begin{corollary}\label{corollary: upper bound of computation-outage probability}
\emph{For large transferred power ($P_0\gg 1$), the computation-outage probability in the case of server-WPT can be bounded as: 
\begin{equation}\label{Eqn: upper-bound of the outage probability in server-WPT}
    P_{\text{out}}'\leq\frac{\tau^L}{\Gamma(L)^2L(1+\alpha L)}\ln{\frac{1}{\tau}} + O\l(P_0^{-L}\r).
\end{equation}
}
\end{corollary}
\begin{proof}
See Appendix \ref{proof: upper bound of computation-outage probability}.
\end{proof}

Given the above result, the  proof of Theorem \ref{theorem: expected learning performance} can be straightforwardly extended to the current case by modifying the expressions of computation-outage probability and harvested energy, yielding the following main result of this section. 

\begin{theorem}\label{theorem: convergence with server-WPT}
(Convergence with Server-WPT). \emph{Consider the case of server-WPT. If the transmission power for WPT to each device is large ($P_0\gg1$), the convergence rate of WP-FEEL is bounded as
\begin{equation}\label{Eqn: convergence rate with server-WPT}
    \mE\l[\frac{1}{N}\sum_{i=0}^{N-1}\|\nabla F(\bw^{(i)})\|^2\r]\leq\frac{2[F(\bw^{(0)})-F_*]}{\eta N}+\frac{\widetilde{\delta} W\overline{C_{\sigma}}}{\rho^{\frac{1}{3}}K T^{\cmp}P_0^{\frac{1}{3}}}+O(P_0^{-L}\ln{P_0}),\quad P_0\to\infty,
\end{equation}
where the constant $\widetilde{\delta}$ is defined as $\widetilde{\delta}=\frac{2B(\frac{2}{3},\frac{1}{6}+\frac{1}{\alpha})\Gamma(L+\frac{1}{3})R^{\frac{\alpha}{3}}}{\alpha\Gamma(L)}$, and the upper-bound of the computation-outage probability $P'_{\text{out}}$ as shown in \eqref{Eqn: upper-bound of the outage probability in server-WPT} scales as $O(P_0^{-L}\ln P_0)$.
}  
\end{theorem}
Comparing Theorems \ref{theorem: expected learning performance} and \ref{theorem: convergence with server-WPT}, the convergence-rate bound for server-WPT has a similar form as that for beacon-WPT except for two differences. First, the spatial energy density in the latter is replaced with transmission energy per-round for WPT, $P_0T^{\cmp}$. Second, the scaling law of the last (residual) term in \eqref{Eqn: convergence rate with server-WPT} differs from that in \eqref{Eqn: convergence rate with PB-WPT} due to WPT-link fading. The effects of other parameters are identical to those discussed in Section \ref{section: optimal learning-WPT tradeoff}. 

\subsection{Optimizing Server-WPT}\label{section: extended version}
The fixed power allocation for WPT in the preceding sub-section can be relaxed to improve the convergence performance. While more sophisticated designs are possible (e.g., involving optimization of  the number of active devices), we consider the following  practical  two-step scheme to be applied in each round (with the index  $i$ omitted in the sequel to simplify notation). 

\begin{itemize}
\item {\bf Step 1} (Scheduling): Let $\tP_k$ denote the power allocated for WPT to device $k$. Considering equal power allocation ($\tP_k = P_0$), select the set of active devices, $\cM'$, by applying the computation-outage criterion in \eqref{Eqn: outage event in server-WPT}. 
\item {\bf Step 2} (Optimal Power Control): Given $\cM'$ and the number of active devices $M' =|\cM'|$, optimize the power allocation  under the sum power constraint $\sum_{k\in\cM'}\tP_k\leq M' P_0$ to minimize sum-local-gradient deviation. 
\end{itemize}

For Step 2, given $\cM'$ and \eqref{Eqn: gradient deviation vs E_k^cmp}, the sum-local-gradient deviation is given as 
\begin{equation}\label{Eqn: 52}
\sum_{k\in \cM'}G_{\lo,k}^{\star}=\frac{W}{(T^{\cmp})^{\frac{2}{3}}}\sum_{k\in \cM'}\frac{\sigma_k^2C_k^{\frac{1}{3}}}{(E_k^{\cmp})^{\frac{1}{3}}}.
\end{equation}
By substituting $\{E_k^{\cmp}\}$ into \eqref{Eqn: 52}, the problem of optimal power control can be formulated as 
\begin{equation*}\label{problem: power allocation}{\bf (P3)}\quad
    \begin{aligned}
        \min_{\{\tP_k\}}\quad &\frac{W}{{(T^{\cmp}})^{\frac{2}{3}}}\sum_{k\in\cM'}\frac{\sigma_k^2C_k^{\frac{1}{3}}}{\l(\frac{\rho \|\tbh_k\|^2\tP_kT^{\cmp}}{r_k^{\alpha}}-\frac{r_k^{\alpha}\varphi(T^{\cmm})}{\|\bh_k\|^2}\r)^{\frac{1}{3}}}\\
        \text{s.t.}\quad 
        &\sum_{k\in\cM'}\tP_k\leq M' P_0. 
    \end{aligned}
\end{equation*}
By a  straightforward application of the KKT conditions, the optimal power-allocation policy is derived in closed-form as: 
\begin{equation}\label{Eqn: power allocation policy}
    \tP_k^\star=\frac{r_k^{2\alpha}\varphi(T^{\cmm})}{\rho \|\tbh_k\|^2\|\bh_k\|^2T^{\cmp}}+\frac{r_k^{\frac{\alpha}{4}}\sigma_k^{\frac{3}{2}}C_k^{\frac{1}{4}}}{\|\tbh_k\|^{\frac{1}{2}}\theta}\l(P_0-\varsigma\r),~~\forall k\in\cM', 
\end{equation}
where 
\begin{equation}\label{Eqn: two parameters}
    \theta=\frac{1}{M'}\sum\limits_{k\in\cM}\frac{r_k^{\frac{\alpha}{4}}\sigma_k^{\frac{3}{2}}C_k^{\frac{1}{4}}}{\|\tbh_k\|^{\frac{1}{2}}}\quad\text{and}\quad \varsigma=\frac{1}{M'}\sum\limits_{k\in\cM}\frac{r_k^{2\alpha}\varphi(T^{\cmm})}{\rho \|\tbh_k\|^2\|\bh_k\|^2T^{\cmp}},
\end{equation}
and thus omitted for brevity. One can observe from \eqref{Eqn: power allocation policy} that the optimal power allocated for WPT sums two components. The first component, which supports gradient uploading, is inversely proportional to the gain of the close-loop channel cascading downlink for WPT and uplink for gradient uploading. The other component, which supports local computation, depends only on the WPT link and is a monotone decreasing function of the channel again of the WPT link. 

\section{Experimental Results}
\subsection{Experimental Settings}
The default  settings are as follows. The edge server is equipped with $64$ antennas and  the cell radius is set as $R=100$ m.  There are $K=30$ edge devices uniformly distributed in the cell, each of which is allocated an uplink bandwidth $B=1$ MHz.  The noise spectrum density is $N_0=-80$ dBm/Hz and the path loss exponents are $\alpha=3.8$ and $\beta=4$. For the WPT path-loss model, $\nu=1$. The processor coefficients $\{C_k\}$ are chosen by uniformly sampling  the set $\{0.010,0.011,\cdots,0.100\}$ [in $\text{Watt}\cdot(\text{MFLOPs}/\text{s})^{-3}$].  The learning  task aims at training a CNN model to classify  handwritten digits using  the well-known MNIST dataset.  For non-i.i.d. data distribution, we first arrange $6\times 10^4$ data samples  according to their labels, follow the sample sequence to divide  the dataset into 60 subsets each  of size 1000, and assign each of 30 devices 2 data subsets. The classifier model is implemented using a 6-layer CNN which consists of two $5\times5$ convolution layers with ReLU activation, each followed by $2\times2$ max pooling, a fully connected layer with $50$ units and ReLU activation, and a final softmax output layer. The total number of parameters is $q=21,840$ and the per-sample computation workload is $W = 1.09\times 10^6$ FLOPs. Furthermore, we suppose that each parameter of the training model gradient is quantized into $Q=16$ bits, and as a result, the transmission overhead in one round is $3.49\times10^5$ bits. We fix the number of rounds as $N=500$ and evaluate  the learning performance in terms of  average gradient norm (over rounds) and  test accuracy.
\begin{figure}[t!]
    \centering
    \subfigure[Average gradient norm]{
    \label{Fig: GDev vs energy density}
    \includegraphics[width=8.89cm]{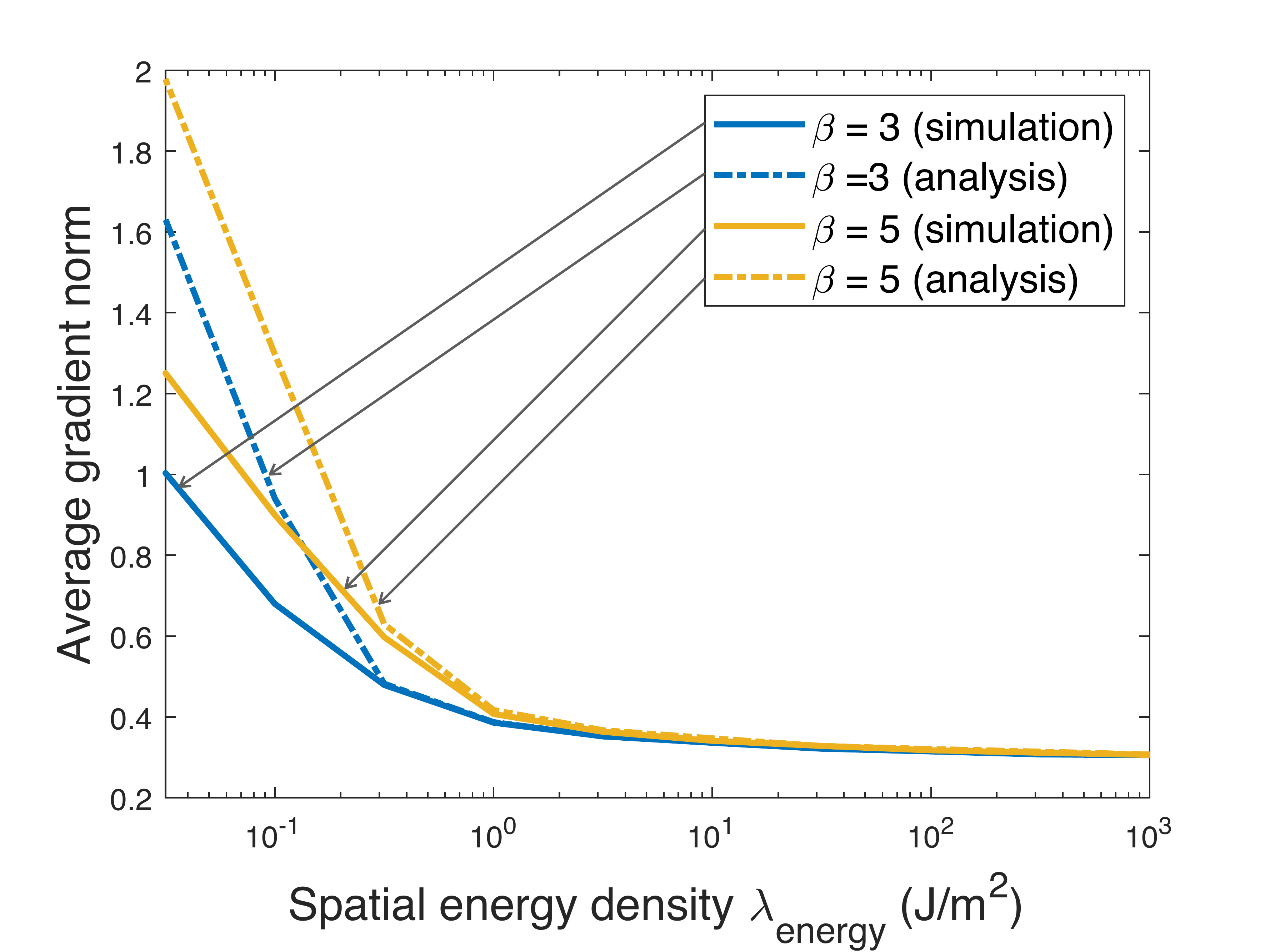}}
    \subfigure[Test accuracy]{
    \label{Fig: accu vs energy density}
    \includegraphics[width=8.89cm]{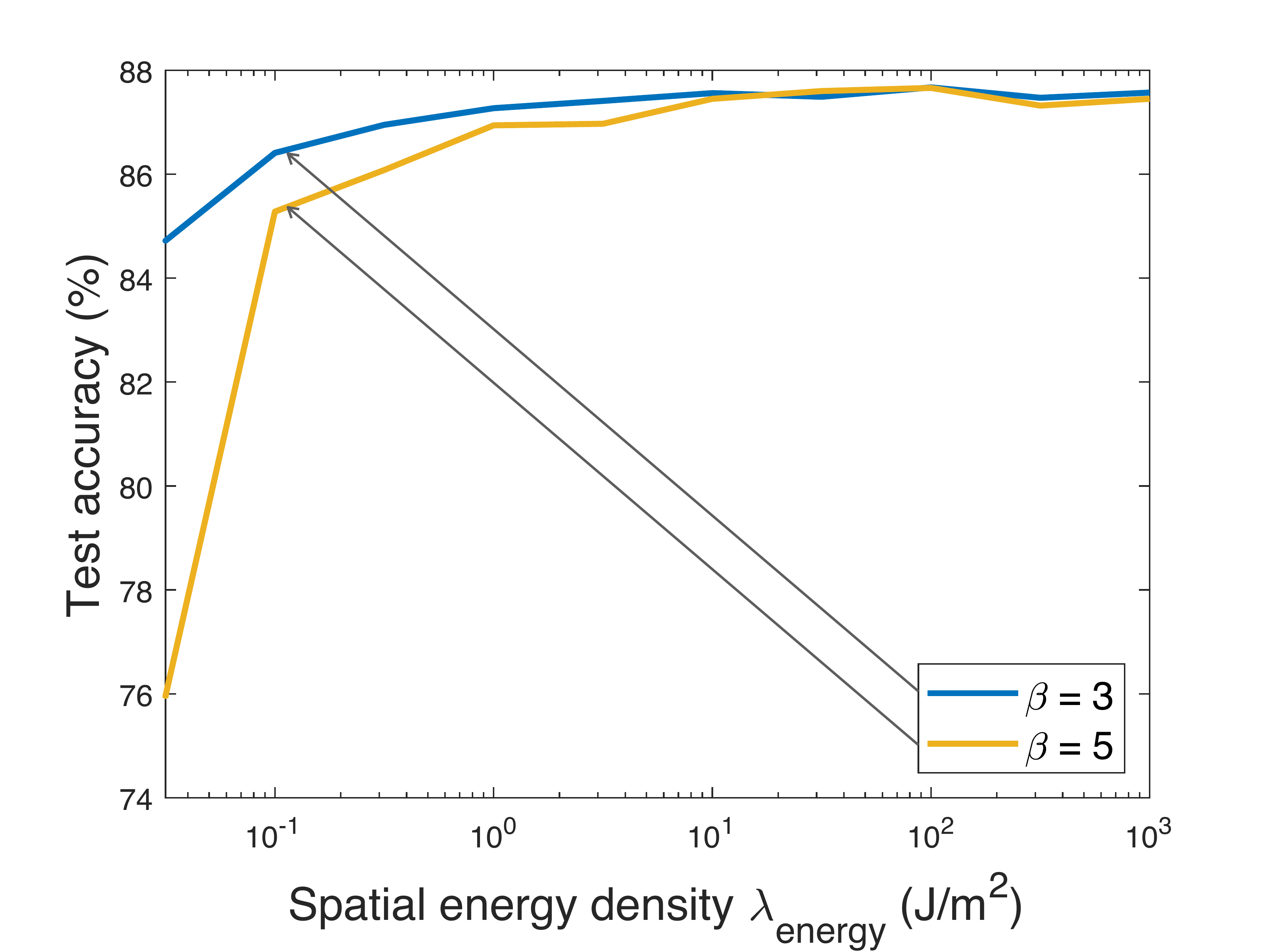}}\vspace{-1mm}
    \caption{Effects of the spatial-energy density on the  performance of a WP-FEEL system with beacon-WPT.}\vspace{-5mm}
    \label{Fig: learning performance vs energy density}
\end{figure}

\subsection{Performance  of WP-FEEL with Beacon-WPT}

The curves of learning performance versus the spatial-energy density $\lambda_{\text{energy}}$ provided by the power-beacon network are plotted in Fig.~\ref{Fig: learning performance vs energy density} for a varying path-loss exponent of WPT links, $\beta$. Both experimental and analytical results are presented. Several observations can be made. As $\lambda_{\text{energy}}$  grows, the increase of transferred energy allows more devices to participate in learning or equivalently  more distributed data to be  exploited for model training. Consequently, one can observe that both the average gradient norm and test accuracy saturate as they converge to their ground-truths. Next, before the convergence, the average gradient norms from analysis and experiments  follow the same scaling laws. That validates the analytical model and the results in Theorem \ref{theorem: expected learning performance} and its corollary. Last, one can observer that a larger value of $\beta$ and hence smaller path loss results in better performance as more energy can be transferred from beacons to devices.

\begin{figure}[t!]
    \centering
    \subfigure[Average gradient norm]{
    \label{Fig: GDev vs computation}
    \includegraphics[width=8.89cm]{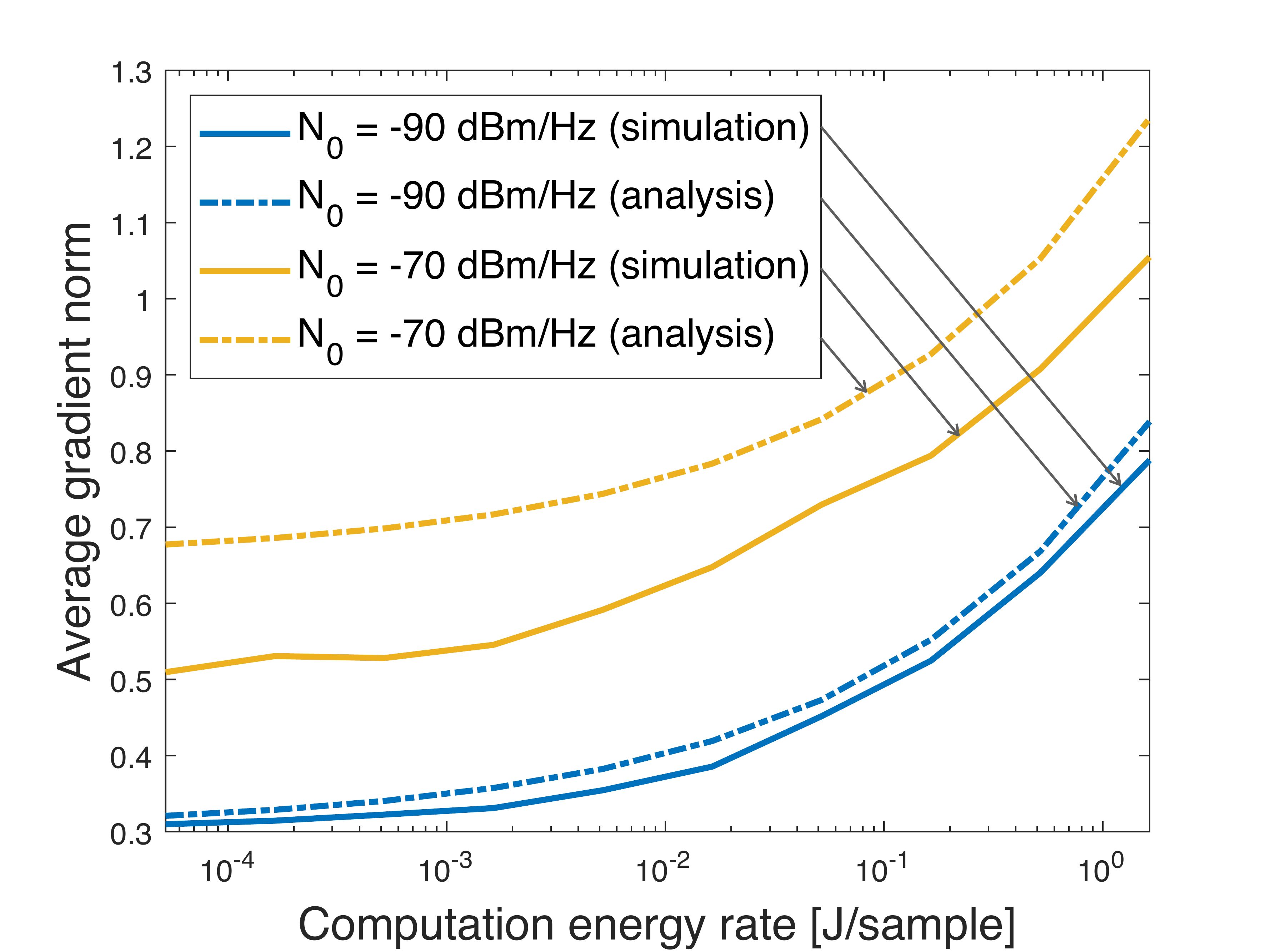}}
    \subfigure[Test accuracy]{
    \label{Fig: accu vs computation}
    \includegraphics[width=8.89cm]{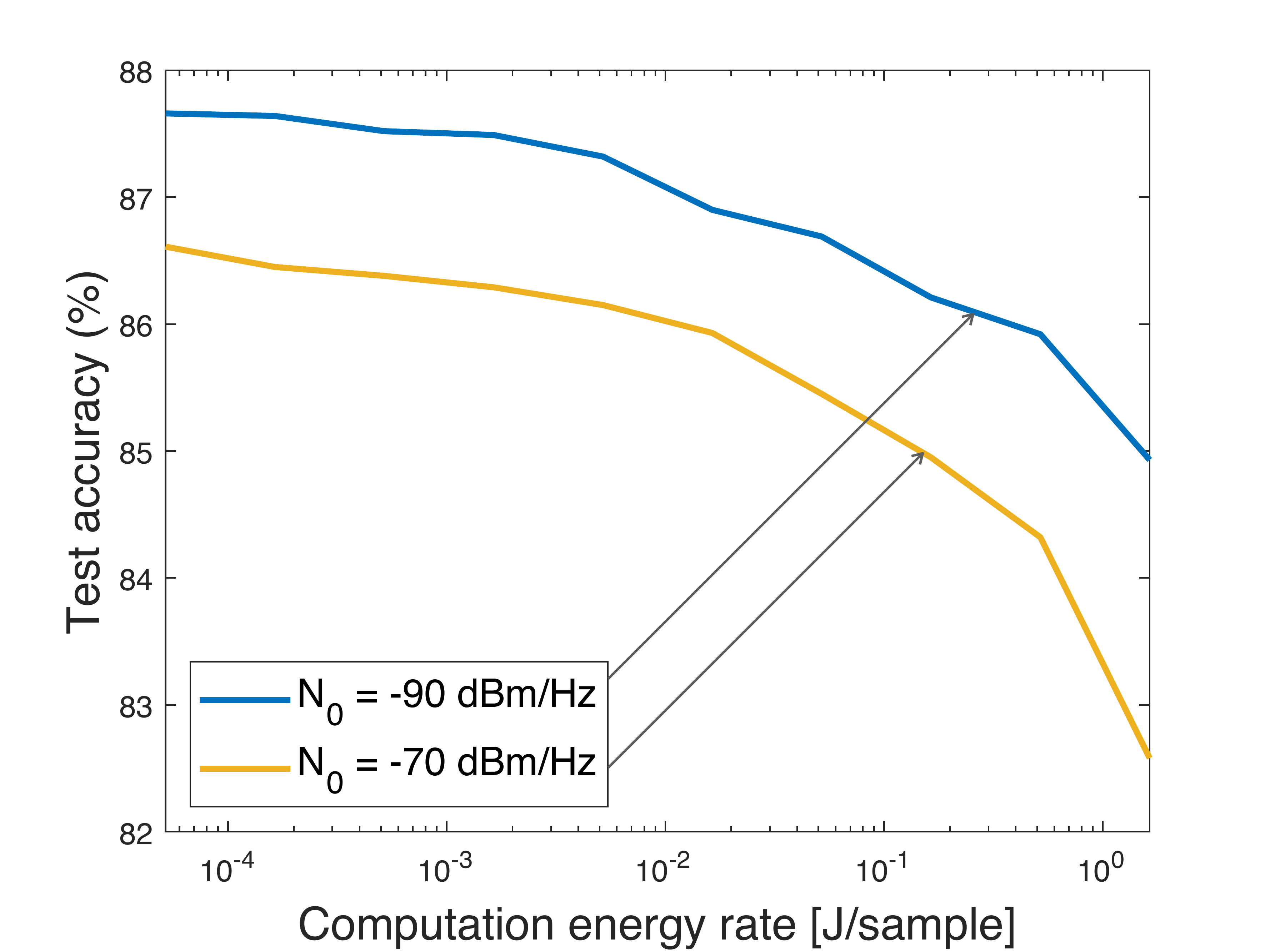}}\vspace{-1mm}
    \caption{Effects of computation energy consumption and channel noise on the performance of a WP-FEEL system with beacon-WPT.}\vspace{-5mm}
    \label{Fig: learning performance vs computation}
\end{figure}

Define the computation-energy rate of a device  as the number of FLOPs computed by its processor per unit energy consumption.  For ease of exposition, consider the case of uniform computation-energy rates for all devices.  The curves of learning performance versus computation-energy rates are plotted  in Fig. \ref{Fig: learning performance vs computation} for a varying noise power spectrum density $N_0$. Several observations can be made. Both the analytical and experiment results are presented in the figure. Their discrepancy arises from some mismatch between the general analytical model specified in the common Assumptions \ref{assumption: smoothness}--\ref{assumption: bounded gradient norm} and the specific dataset (i.e., MNIST) used in the experiments. The mismatch is observed to less for smaller $N_0$ due to the averaging effect of more devices/more data involved in learning. Next, aligned with intuition, reducing $N_0$ improves learning performance by activating more devices as well as reducing the communication-energy consumption; thereby more energy is allowed for local gradient estimation and its accuracy  improves. Next, given fixed harvested energy per device,  one can observe degradation of learning performance as the increasing computation-energy rate reduces the mini-batch size, the number of active devices, and the energy available for communication.

\begin{figure}[t!]
    \centering
    \subfigure[Cell radius  $R=20$ m.]{
    \label{Fig: R=20}
    \includegraphics[width=8.3cm]{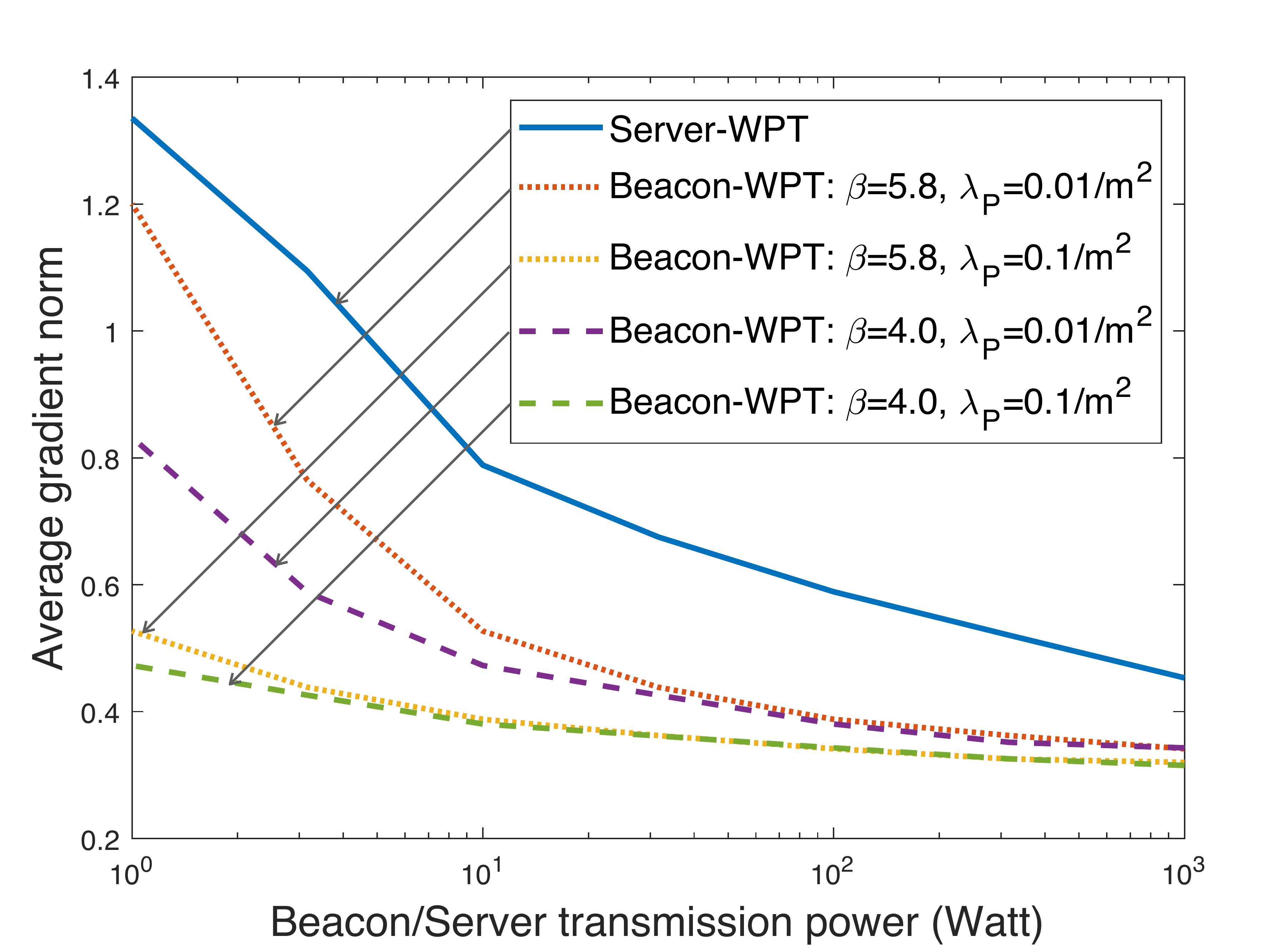}
    \includegraphics[width=8.3cm]{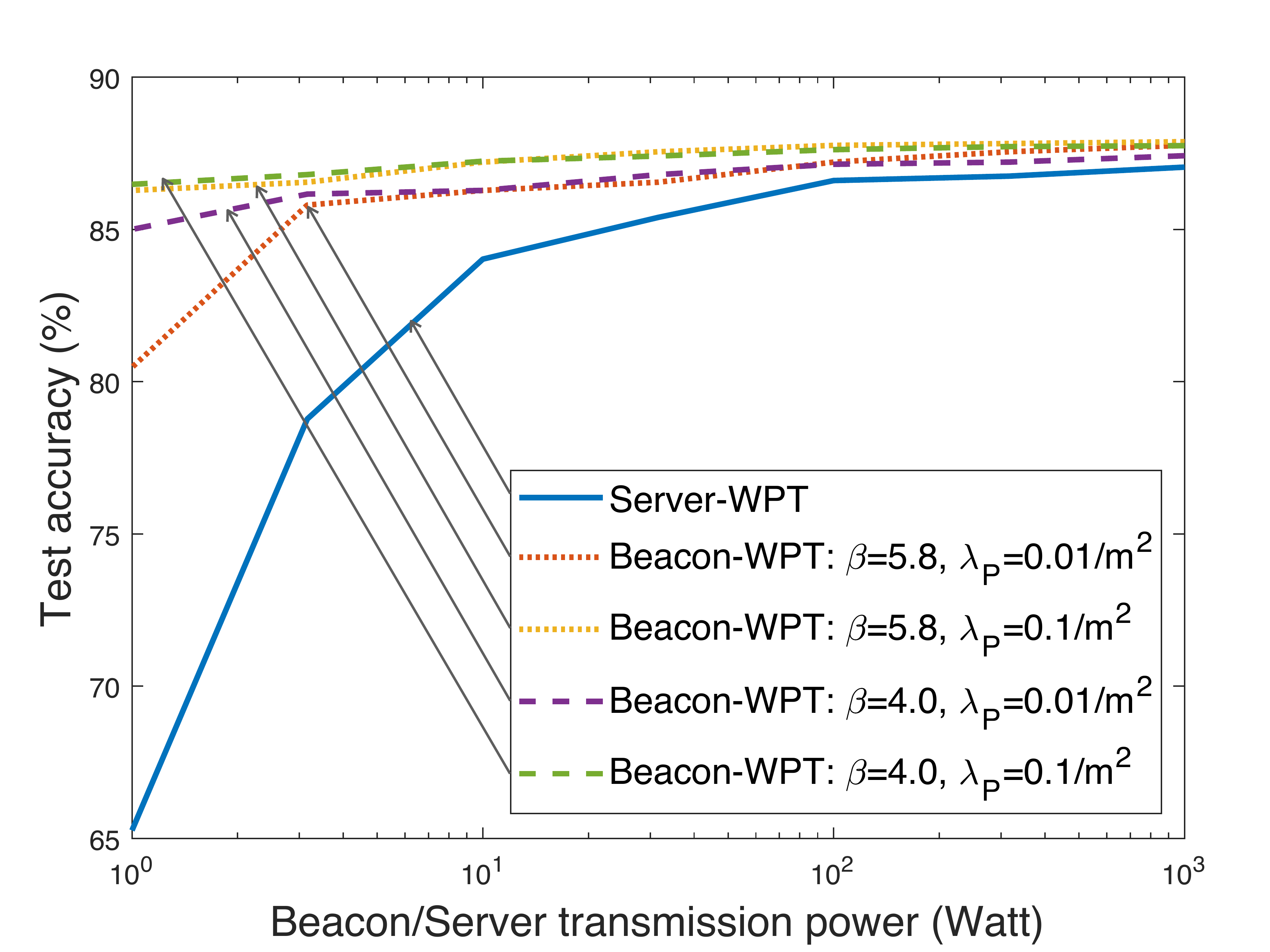}}\vspace{-1mm}
    \hspace{-2mm}
    \subfigure[Cell radius  $R=10$ m.]{
    \label{Fig: R=10}
    \includegraphics[width=8.3cm]{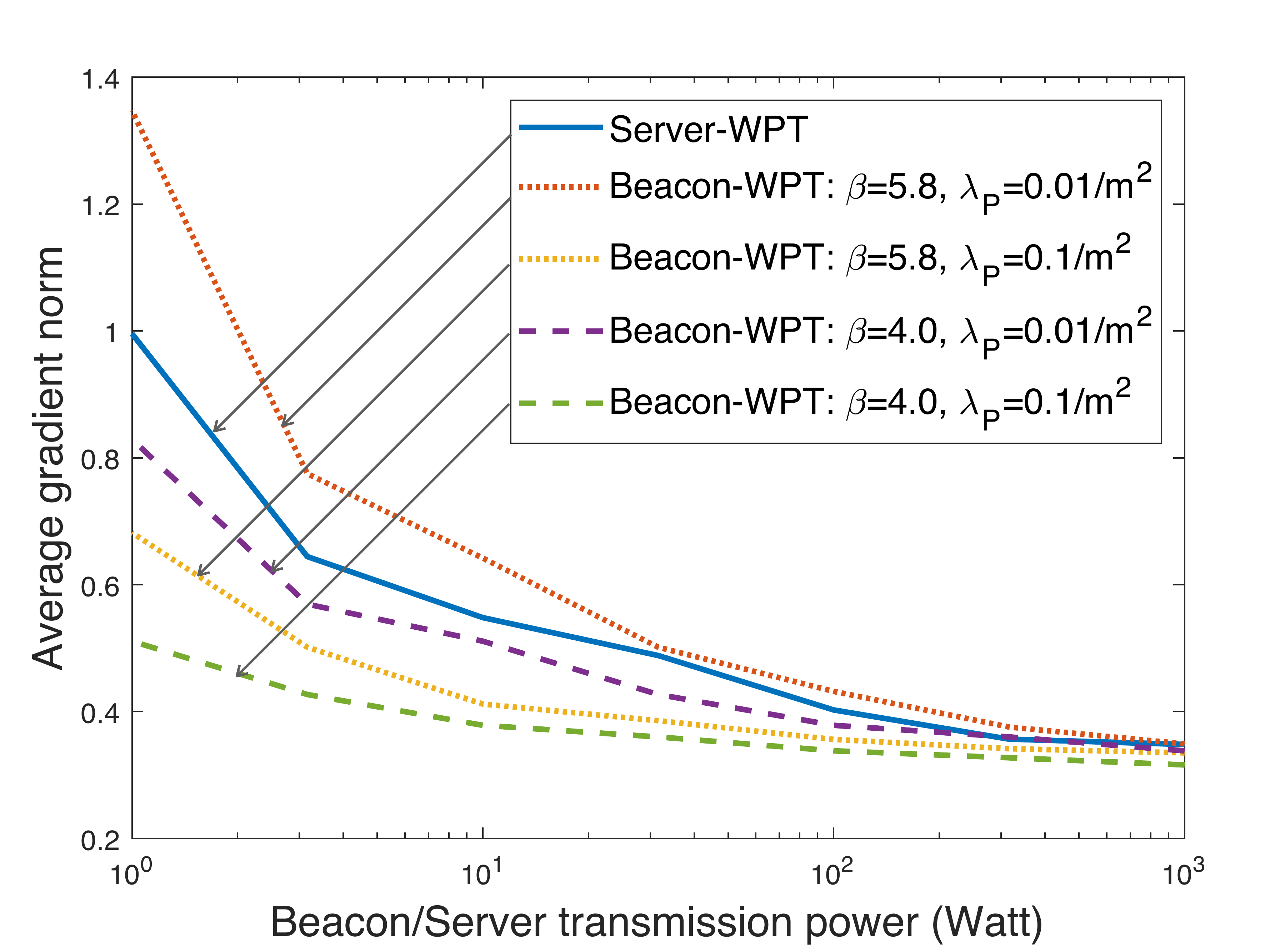}
    \includegraphics[width=8.3cm]{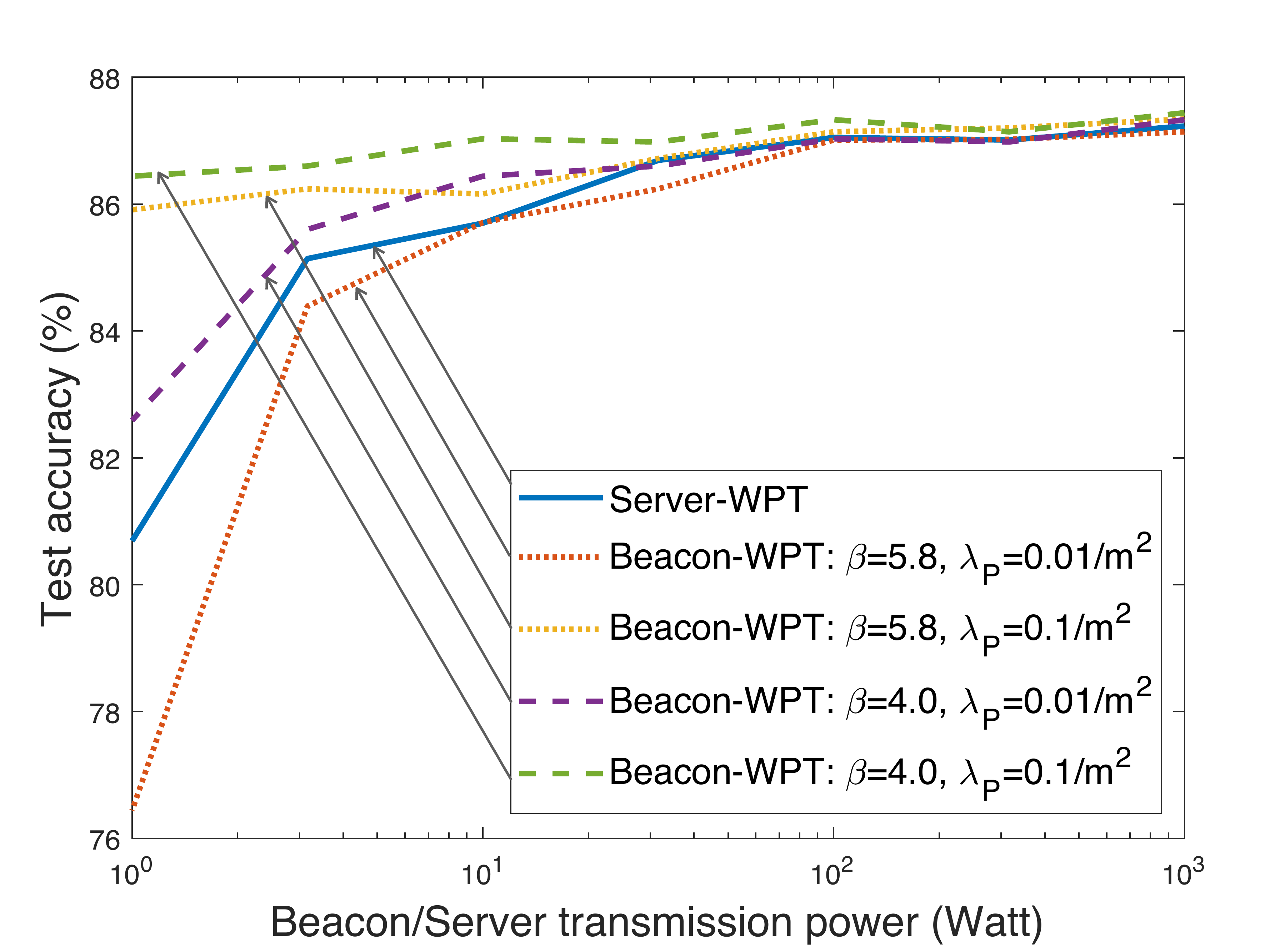}}\vspace{-1mm}
    \caption{Learning-performance comparison between beacon-WPT and server-WPT.}\vspace{-5mm}
    \label{Fig: comparison}
\end{figure}
\subsection{Comparison between Beacon-WPT and Server-WPT}
Based on experiments, the learning performance for the scenarios of beacon-WPT and server-WPT is compared  in Fig. \ref{Fig: comparison} for two different cell sizes. For the purpose of comparison,  the transmission power of power-beacons and serve are equalized $\bar{P}=P_0$ while the beacon density, $\lambda_{\text{pb}}$,  and the path-loss exponent of beacon-WPT links, $\beta$,  are varied. For a  relatively large cell ($R = 20$ m), it can be observed that beacon-WPT always outperforms server-WPT in the considered ranges of settings as the latter with only a single power source suffers from a low WPT efficiency due to  severe path loss. This suggests the need of deploying power-beacons to power devices if FEEL is deployed in a large area to involve many devices. On the other hand, for a relatively small cell ($R = 10$ m), server-WPT has a higher WPT efficiency and can outperform beacon-WPT when beacons are sparse  and/or the path-loss exponent $\beta$ is large.


\section{Concluding Remarks}
In this paper, we have proposed the application of WPT to a FEEL system as a solution for the practical issue of high energy consumption at devices. To study the performance of the resultant WP-FEEL system, we have analyzed the optimal learning-WPT tradeoffs for both the scenarios of beacon-WPT and server-WPT. The results contribute useful insight and algorithms for designing and deploying WP-FEEL systems. This first study of WP-FEEL opens several directions for further research including the application of WPT to support the implementation of other edge-learning frameworks (e.g., parameter server and reinforcement learning), the use of more complex  wireless techniques (e.g., over-the-air aggregation and radio resource management), and  the design of multi-cell networks.

\appendix


\subsection{Proof of Proposition \ref{proposition: convergence rate}}\label{proof: convergence rate}
According to the update rule in \eqref{eqn: update rule} and Assumption \ref{assumption: smoothness}, we have
\begin{align}\label{A2}
    F(\bw^{(i+1)})-F(\bw^{(i)})&\leq \langle \nabla F(\bw^{(i)}),\bw^{(i+1)}-\bw^{(i)}\rangle +\frac{\mu}{2}\|\bw^{(i+1)}-\bw^{(i)}\|^2\nonumber\\
    &=-\eta\langle\nabla F(\bw^{(i)}),\bg^{(i)}\rangle+\frac{\mu\eta^2}{2}\|\bg^{(i)}\|^2.
\end{align}
Using $\|\bg^{(i)}\|^2=\|\bg^{(i)}-\nabla F(\bw^{(i)})\|^2-\|\nabla F(\bw^{(i)})\|^2+2\langle\nabla F(\bw^{(i)}),\bg^{(i)}\rangle$, we can derive
\begin{align}\label{B1}
    F(\bw^{(i+1)})-F(\bw^{(i)})&\leq\frac{\mu\eta^2}{2}\l(\|\bg^{(i)}-\nabla F(\bw^{(i)})\|^2-\|\nabla F(\bw^{(i)})\|^2\r)+(\mu\eta-1)\eta\langle\nabla F(\bw^{(i)}),\bg^{(i)}\rangle\nonumber\\
    &=\frac{\mu\eta^2}{2}\|\bg^{(i)}-\nabla F(\bw^{(i)})\|^2+\l(\frac{\mu\eta}{2}-1\r)\eta\|\nabla F(\bw^{(i)})\|^2\nonumber\\
    &\quad -(1-\mu\eta)\eta\langle\nabla F(\bw^{(i)}),\bg^{(i)}-\nabla F(\bw^{(i)})\rangle.
\end{align}
The third term at the right-hand side in \eqref{B1} can be upper bounded as
\begin{equation}\label{B2}
    -\langle\nabla F(\bw^{(i)}),\bg^{(i)}-\nabla F(\bw^{(i)})\rangle\leq\frac{1}{2}\l(\|\nabla F(\bw^{(i)})\|^2+\|\bg^{(i)}-\nabla F(\bw^{(i)})\|^2\r).
\end{equation}
Given $0<\eta\leq\frac{1}{\mu}$, we substitute (\ref{B2}) into (\ref{B1}) and rearrange the result, yielding
\begin{equation}
    \|\nabla F(\bw^{(i)})\|^2\leq\frac{2\l(F(\bw^{(i)})-F(\bw^{(i+1)})\r)}{\eta}+\|\bg^{(i)}-\nabla F(\bw^{(i)})\|^2.
\end{equation}
It follows that
\begin{align}\label{A5}
    \mE\l[\frac{1}{N}\sum_{i=0}^{N-1}\|\nabla F(\bw^{(i)})\|^2\r]&\leq\frac{2\l(F(\bw^{(0)})-\mE \l[F(\bw^{(N)})\r]\r)}{\eta N} +\frac{1}{N}\sum_{i=0}^{N-1}\mE\l[\|\bg^{(i)}-\nabla F(\bw^{(i)})\|^2\r]\nonumber\\
    &\leq \frac{2\l(F(\bw^{(0)})-F_*\r)}{\eta N} +\frac{1}{N}\sum_{i=0}^{N-1}\mE\l[\|\bg^{(i)}-\nabla F(\bw^{(i)})\|^2\r],
\end{align}
where the second inequality in \eqref{A5} follows Assumption \ref{assumption: bounded loss function}. This completes the proof.

\subsection{Proof of Lemma~\ref{proposition: global gradient deviation}}\label{proof: global gradient deviation}
The aggregated global gradient can be written using the indicator function $\mI$ as
\begin{equation}
    \bg=\frac{1}{M}\sum\nolimits_{k\in\cM}\bg_k=\frac{1}{M}\sum\nolimits_{k=1}^K\mI_{k\in\cM}\bg_k,~\text{if}~M>0; ~\text{and}~\bg=\mathbf{0},~\text{otherwise}.
\end{equation}
For tractability, we introduce the following auxiliary gradient $\tbg$ which is defined by
\begin{equation}
    \tbg=\frac{1}{K}\sum_{k=1}^K\mI_{k\in\cM}\bg_k,~\text{where}~\mI_{k\in\cM}\bg_k=0~\text{for}~k\in\cK\setminus\cM.
\end{equation}
Then, we can derive the following upper bound for $G_{\gl}$ using the auxiliary gradient:
\begin{equation}\label{C1}
\begin{aligned}
    G_{\gl}=\mE_{\cM}\l\{\mE\l[\|\tbg-\nabla F(\bw)+\bg-\tbg\|^2\r]\r\}\leq 2\mE_{\cM}\Big\{\underbrace{\mE\l[\|\tbg-\nabla F(\bw)\|^2\r]}_{\text{(a)}}+\underbrace{\mE\l[\|\bg-\tbg\|^2\r]}_{\text{(b)}}\Big\}.
\end{aligned} 
\end{equation}
In the following, we focus on terms (a) and (b) defined in \eqref{C1}, respectively.
\begin{itemize}
    \item[1) ] Since $\nabla F(\bw)=\frac{1}{K}\sum_{k=1}^{K}\nabla F_k(\bw)$, we can first decompose term (a) as follows:
    \begin{align*}
        \text{(a)}&=\mE\l[\|\tbg-\nabla F(\bw)\|^2\r]=\frac{1}{K^2}\sum_{k=1}^K\sum_{\ell=1}^K\mE\l[\langle\mI_{k\in\cM}\bg_k-\nabla F_k(\bw),\mI_{\ell\in\cM}\bg_{\ell}-\nabla F_{\ell}(\bw)\rangle\r]\nonumber\\
        &=\frac{1}{K^2}\bigg(\underbrace{\sum_{k=1}^K\mE[\|\mI_{k\in\cM}\bg_k-\nabla F_k(\bw)\|^2]}_{\text{(a1)}}+\underbrace{\sum_{k\neq\ell}\mE[\langle\mI_{k\in\cM}\bg_k-\nabla F_k(\bw),\mI_{\ell\in\cM}\bg_{\ell}-\nabla F_{\ell}(\bw)\rangle]}_{\text{(a2)}}\bigg),
    \end{align*}
    Next, we aim at finding the upper bounds for terms (a1) and (a2), respectively.
    \begin{itemize}
        \item[$\bullet$] For term (a1), we can bound it as follows:
        \begin{align}
            \text{(a1)}&=\sum_{k=1}^K\mI_{k\in\cM}\mE\l[\|\bg_k-\nabla F_k(\bw)\|^2\r]+\sum_{k=1}^K\mI_{k\in\cK\setminus\cM}\l\{\mE\l[\|\nabla F_k(\bw)\|^2\r]\r\}\nonumber\\
            &\leq\sum_{k=1}^K\mI_{k\in\cM}G_{\lo,k}+(K-M)\Phi,
        \end{align} 
        which comes from Jensen's inequality $\|\nabla F_k(\bw)\|^2=\|\mE[\bg_k]\|^2\leq\mE[\|\bg_k\|^2]\leq\Phi$. Then,
        \begin{equation}
        \begin{aligned}
            \mE_{\cM}\l[\text{(a1)}\r]&\leq\mE\l[\sum_{k\in\cM}G_{\lo,k}\r]+(K-\mE[M])\Phi.
        \end{aligned} 
        \end{equation}
        \item[$\bullet$] For term (a2), we can first divide  $\mE\l[\langle\mI_{k\in\cM}\bg_k-\nabla F_k(\bw),\mI_{\ell\in\cM}\bg_{\ell}-\nabla F_{\ell}(\bw)\rangle\r]$ into three categories $(k\neq\ell)$:
        \begin{enumerate}
            \item[i)] Case 1: $k\in\cM$ and $\ell\in\cM$. Given Assumption \ref{assumption: local gradients estimation}, $\bg_k$ and $\bg_{\ell}$ are uncorrelated, so 
            \begin{align}
                &\mE\l[\langle\mI_{k\in\cM}\bg_k-\nabla F_k(\bw),\mI_{\ell\in\cM}\bg_{\ell}-\nabla F_{\ell}(\bw)\rangle\r]\nonumber\\
                &=\mE\l[(\bg_k-\nabla F_k(\bw))^{\mathsf{T}}(\bg_{\ell}-\nabla F_{\ell}(\bw))\r]=\tr\l\{\mathsf{cov}(\bg_k,\bg_{\ell})\r\}=0.
            \end{align}
            \item[ii)] Case 2: $k\in\cM$ but $\ell\notin\cM$ (or $k\notin\cM$ but $\ell\in\cM$). Given Assumption \ref{assumption: local gradients estimation}, $\bg_k$ is an unbiased estimate of $\nabla F_k(\bw)$, resulting in
            \begin{equation}
            \begin{aligned}
                \hspace{-9mm}\mE[\langle\mI_{k\in\cM}\bg_k-\nabla F_k(\bw),\mI_{\ell\in\cM}\bg_{\ell}-\nabla F_{\ell}(\bw)\rangle]
                =-\mE[(\bg_k-\nabla F_k(\bw))^{\mathsf{T}}\nabla F_{\ell}(\bw)]=0.\!\!
            \end{aligned}
            \end{equation}
            \item[iii)] Case 3: $k\notin\cM$ and $\ell\notin\cM$. When both device $k$ and $\ell$ are outage, it holds
            \begin{equation}
            \begin{aligned}
                \mE\l[\langle\mI_{k\in\cM}\bg_k-\nabla F_k(\bw),\mI_{\ell\in\cM}\bg_{\ell}-\nabla F_{\ell}(\bw)\rangle\r]=\nabla F_k(\bw)^{\mathsf{T}}\nabla F_{\ell}(\bw).
            \end{aligned}
            \end{equation}
        \end{enumerate}
        Combining the above three cases, we can bound term (a2) as follows:
        \begin{align}
            \text{(a2)}&=\sum_{k\neq\ell}\mI_{k\notin\cM}\mI_{\ell\notin\cM}\langle\nabla F_k(\bw),\nabla F_{\ell}(\bw)\rangle\nonumber\\
            &\leq\sum_{k\neq\ell}\mI_{k\notin\cM}\mI_{\ell\notin\cM}\|\nabla F_k(\bw)\|\|\nabla F_{\ell}(\bw)\|
            \leq\sum_{k\neq\ell}\mI_{k\notin\cM}\mI_{\ell\notin\cM}\Phi.
        \end{align}
        Then, taking expectation over $\cM$, we can obtain
        \begin{equation}
        \begin{aligned}
            \mE_{\cM}\l[\text{(a2)}\r]&
            \leq\mE\l[\sum_{k\neq\ell}\mI_{k\notin\cM}\mI_{\ell\notin\cM}\r]\Phi=(K^2-K)P_{\text{out}}^2\Phi.
        \end{aligned}
        \end{equation}
    \end{itemize}
    In summary, the expected term (a) in \eqref{C1} can be upper bounded by
    \begin{align}\label{Eqn: expected (a)}
        \mE_{\cM}\l[\text{(a)}\r]\leq\frac{1}{K^2}\mE\l[\sum_{k\in\cM}G_{\lo,k}\r]+\frac{K-\mE[M]}{K^2}\Phi+\frac{K-1}{K}P_{\text{out}}^2\Phi.
    \end{align} 
    \item[2) ] For term (b) in \eqref{C1}, we can derive the upper bound as follows:
    \begin{align}
        \text{(b)}&=\mE\l[\l\|\bg-\tbg\r\|^2\r]=\mI_{M>0}\mE\l[\Big\|\l(\frac{1}{M}-\frac{1}{K}\r)\sum_{k\in\cM}\bg_k\Big\|^2\r]+\mI_{M=0}\times0\nonumber\\
        &\leq\mI_{M>0}\l(\frac{1}{M}-\frac{1}{K}\r)^2\sum_{k\in\cM}\mE\l[\|\bg_k\|^2\r]\leq \mI_{M>0}\l(\frac{1}{M}+\frac{M}{K^2}-\frac{2}{K}\r)\Phi,
    \end{align} 
    where we note that, when $M=0$, both $\bg$ and $\tbg$ are zero, thus $\mE\l[\|\bg-\tbg\|^2|M=0\r]=0$.
    Furthermore, considering $$\mE[M]=\Pr(M>0)\mE[M|M>0]+\Pr(M=0)\mE[M|M=0]=\Pr(M>0)\mE[M|M>0]$$ and $\Pr(M>0)=1-P_{\text{out}}^K$, we can bound the expected (b) as
    \begin{align}\label{Eqn: expected (b)}
        \mE_{\cM}[\text{(b)}]&\leq\Pr(M>0)\mE\l[\frac{1}{M}+\frac{M}{K}-\frac{2}{K}\Big|M>0\r]\Phi\nonumber\\
        &=\l((1-P_{\text{out}}^K)\mE\l[\frac{1}{M}\Big|M>0\r]+\frac{1}{K^2}\mE\l[M\r]-(1-P_{\text{out}}^K)\frac{2}{K}\r)\Phi.
    \end{align}
\end{itemize}
Finally, substituting the results \eqref{Eqn: expected (a)} and \eqref{Eqn: expected (b)} into \eqref{C1}, we have
\begin{align}
    \frac{1}{2}G_{\gl}&\leq\frac{1}{K^2}\mE\l[\sum_{k\in\cM}G_{\lo,k}\r]+(1- P_{\text{out}}^K)\l(\mE\l[\frac{1}{M}\Big|M>0\r]-\frac{1}{K}\r)\Phi+\frac{P_{\text{out}}^K- P_{\text{out}}^2}{K}+ P_{\text{out}}^2\Phi\nonumber\\
    &\leq\frac{1}{K^2}\mE\l[\sum_{k\in\cM}G_{\lo,k}\r]+(1- P_{\text{out}}^K)\l(\mE\l[\frac{1}{M}\Big|M>0\r]-\frac{1}{K}\r)\Phi+ P_{\text{out}}^2\Phi,
\end{align}
where we note that $P_{\text{out}}^K- P_{\text{out}}^2\leq0$ due to $K\geq2$. This completes the proof.
\subsection{Proof of Lemma \ref{lemma: outage probability}}\label{proof: outage probability}
According to Definition \ref{definition: outage event}, in the $i$-th round, the device $k$ encounters a computation-outage event if the condition $E_k^{\cmm}(r_k^{(i)},\bh_k^{(i)})\geq\bar{E}$, which is equivalent to $R^{\alpha}\|\bh_k^{(i)}\|^2\leq \Big(r_k^{(i)}\Big)^{\alpha}\frac{\phi(T^{\cmm})}{R^{-\alpha}\bar{E}}$, is true. Define the domain $\Xi=\{(r,h)\in\mR^2:R^{\alpha}h\leq r^{\alpha}\xi;~h\geq0;~0\leq r\leq R\}$ for inactive devices, where $\xi$ has the same definition as mentioned in \eqref{Eqn: xi}. Since $\bh_k^{(i)}$ and $r_k^{(i)}$ are independent, we know $f_{(r_k^{(i)},\|\bh_k^{(i)}\|^2)}(r,h)=f_{\|\bh_k^{(i)}\|^2}(h)f_{r_k^{(i)}}(r)$. Then, we can derive the outage probability:
\begin{align}
    P_{\text{out}}&=\iint_{\Xi}f_{\|\bh_k^{(i)}\|^2}(h)f_{r_k^{(i)}}(r)drdh=\int_0^R\frac{2r}{R^2}\int_0^{\l(\frac{r}{R}\r)^{\alpha}\xi}\frac{h^{L-1}e^{-h}}{\Gamma(L)}dhdr\nonumber\\
    &=\frac{2}{R^2\Gamma(L)}\int_0^R\gamma\l(L,\l(\frac{r}{R}\r)^{\alpha}\xi\r)rdr=\frac{\gamma(L,\xi)-\xi^{-\frac{2}{\alpha}}\gamma\l(L+\frac{2}{\alpha},\xi\r)}{\Gamma(L)}.
\end{align}
This completes the proof.
\subsection{Proof of Lemma \ref{lemma: optimal local-gradient deviation}}\label{proof: upper bound for aggregated local gradient deviations}
The aggregated local gradient deviations come from the active devices. Define the domain $\Theta=\{(r,h)\in\mR^2:R^{\alpha}h>r^{\alpha}\xi;~h\geq0;~0\leq r\leq R\}$ for active devices, where $\xi$ has the same definition as mentioned in \eqref{Eqn: xi}. By variable transform, we introduce a new integral variant $x\triangleq\l(\frac{r}{R}\r)^{\alpha}\xi$, so that $r=R\xi^{-\frac{1}{\alpha}}x^{\frac{1}{\alpha}}$ and $dr=\frac{1}{\alpha}R\xi^{-\frac{1}{\alpha}}x^{\frac{1}{\alpha}-1}dx$. Accordingly, the integral domain becomes $\Theta=\{(x,h)\in\mR^2:h>x;~0\leq x\leq\xi\}$. Then we can obtain
\begin{align}\label{D2}
    \mE\l[\mI_{k\in\cM}G_{\lo,k}^{\star}\r]
    &=\frac{W\sigma_k^2C_k^{\frac{1}{3}}\xi^{\frac{1}{3}}}{{T^{\cmp}}^{\frac{2}{3}}\varphi(T^{\cmm})^{\frac{1}{3}}}\iint_{\Theta}\l(\frac{h}{R^{\alpha}h-r^{\alpha}\xi}\r)^{\frac{1}{3}}\frac{h^{L-1}e^{-h}}{\Gamma(L)}\frac{2r}{R^2}drdh\nonumber\\
    &=\frac{2W\sigma_k^2C_k^{\frac{1}{3}}\xi^{\frac{1}{3}-\frac{2}{\alpha}}}{\alpha{T^{\cmp}}^{\frac{2}{3}}R^{\frac{\alpha}{3}}\varphi(T^{\cmm})^{\frac{1}{3}}\Gamma(L)}\underbrace{\iint_{\Theta}\frac{h^{L-\frac{2}{3}}e^{-h}}{(h-x)^{\frac{1}{3}}x^{1-\frac{2}{\alpha}}}dxdh}_{\text{(c)}}.
\end{align}
The integration (c) defined in \eqref{D2} can be decomposed into two terms as follows:
\begin{align}\label{E1}
    \text{(c)}=\underbrace{\int_0^{\xi}\int_0^h\frac{h^{L-\frac{2}{3}}e^{-h}}{(h-x)^{\frac{1}{3}}x^{1-\frac{2}{\alpha}}}dxdh}_{\text{(c1)}}+\underbrace{\int_{\xi}^{\infty}\int_0^{\xi}\frac{h^{L-\frac{2}{3}}e^{-h}}{(h-x)^{\frac{1}{3}}x^{1-\frac{2}{\alpha}}}dxdh}_{\text{(c2)}},
\end{align}
\begin{itemize}
    \item[1) ] For term (c1), we can derive the result as follows:
    \begin{equation}\label{E2}
        \text{(c1)}=B\l(\frac{2}{3},\frac{2}{\alpha}\r)\int_0^{\xi}h^{L+\frac{2}{\alpha}-1}e^{-h}dh=B\l(\frac{2}{3},\frac{2}{\alpha}\r)\gamma\l(L+\frac{2}{\alpha},\xi\r).
    \end{equation}
    \item[2) ] For term (c2), we can find an upper bound as follows:
    \begin{align}\label{E3}
        \text{(c2)}&=\int_{\xi}^{\infty}h^{L-1}e^{-h}\int_0^{\xi}\l(1-\frac{x}{h}\r)^{-\frac{1}{3}}x^{\frac{2}{\alpha}-1}dxdh\nonumber\\
        &\leq\int_{\xi}^{\infty}h^{L-1}e^{-h}dh\int_0^{\xi}\l(1-\frac{x}{\xi}\r)^{-\frac{1}{3}}x^{\frac{2}{\alpha}-1}dx=B\l(\frac{2}{3},\frac{2}{\alpha}\r)\xi^{\frac{2}{\alpha}}\Gamma(L,\xi).
    \end{align}
\end{itemize}
Substituting the results \eqref{E2} and \eqref{E3} into \eqref{D2}, we can derive the following upper bound:
\begin{equation}\label{E4}
    \mE\l[\mI_{k\in\cM}G_{\lo,k}^{\star}\r]\leq\frac{2W\sigma_k^2C_k^{\frac{1}{3}}B(\frac{2}{3},\frac{2}{\alpha})\xi^{\frac{1}{3}}}{\alpha{T^{\cmp}}^{\frac{2}{3}}R^{\frac{\alpha}{3}}\varphi(T^{\cmm})^{\frac{1}{3}}\Gamma(L)}\l(\xi^{-\frac{2}{\alpha}}\gamma\l(L+\frac{2}{\alpha},\xi\r)+\Gamma(L,\xi)\r).
\end{equation}
Expanding $\mE[\mI_{k\in\cM}G_{\lo,k}]$ and comparing the right-hand side with the expression of $P_{\text{out}}$, we have
\begin{equation}
    \Pr(k\in\cM)\mE\l[G_{\lo,k}^{\star}|k\in\cM\r]\leq\frac{2W\sigma_k^2C_k^{\frac{1}{3}}B(\frac{2}{3},\frac{2}{\alpha})\xi^{\frac{1}{3}}}{\alpha{T^{\cmp}}^{\frac{2}{3}}R^{\frac{\alpha}{3}}\varphi(T^{\cmm})^{\frac{1}{3}}}(1-P_{\text{out}}).
\end{equation}
Note that $\Pr(k\in\cM)=1-P_{\text{out}}$ and this completes the proof.

\subsection{Proof of Corollary \ref{corollary: upper bound of computation-outage probability}}\label{proof: upper bound of computation-outage probability}
Applying the second mean value theorem for definite integrals, we can derive (with $\tau\ll1$)
\begin{align}
P_{\text{out}}'&=\int_0^{\tau}\frac{2\mathsf{K}_0(2\sqrt{x})}{\ln{x}}\frac{x^{L-1}\l[1-(\frac{x}{\tau})^{\frac{1}{\alpha}}\r]\ln{x}}{\Gamma(L)^2}dx\nonumber\\
&=\l[\lim_{x\to0^+}\frac{2\mathsf{K}_0(2\sqrt{x})}{\ln{x}}\r]\int_0^{\tau'}\frac{x^{L-1}\l[1-(\frac{x}{\tau})^{\frac{1}{\alpha}}\r]\ln{x}}{\Gamma(L)^2}dx\nonumber\\
&\leq -\int_0^{\tau}\frac{x^{L-1}\l[1-(\frac{x}{\tau})^{\frac{1}{\alpha}}\r]\ln{x}}{\Gamma(L)^2}dx=\frac{\tau^L\l(1+2\alpha L-L(1+\alpha L)\ln{\tau}\r)}{\Gamma(L)^2L^2(1+\alpha L)^2},
\end{align}
where $\frac{2\mathsf{K}_0(2\sqrt{x})}{\ln{x}}$ is a negative monotone increasing function around $0^+$ with $\lim\limits_{x\to0^+}\frac{2\mathsf{K}_0(2\sqrt{x})}{\ln{x}}=-1$ and the number $\tau'\in(0,\tau]$. This completes the proof.

\bibliography{WPT.bib}

\begin{thebibliography}{10}
\providecommand{\url}[1]{#1}
\csname url@samestyle\endcsname
\providecommand{\newblock}{\relax}
\providecommand{\bibinfo}[2]{#2}
\providecommand{\BIBentrySTDinterwordspacing}{\spaceskip=0pt\relax}
\providecommand{\BIBentryALTinterwordstretchfactor}{4}
\providecommand{\BIBentryALTinterwordspacing}{\spaceskip=\fontdimen2\font plus
\BIBentryALTinterwordstretchfactor\fontdimen3\font minus
  \fontdimen4\font\relax}
\providecommand{\BIBforeignlanguage}[2]{{%
\expandafter\ifx\csname l@#1\endcsname\relax
\typeout{** WARNING: IEEEtran.bst: No hyphenation pattern has been}%
\typeout{** loaded for the language `#1'. Using the pattern for}%
\typeout{** the default language instead.}%
\else
\language=\csname l@#1\endcsname
\fi
#2}}
\providecommand{\BIBdecl}{\relax}
\BIBdecl

\bibitem{xu2019edgeAI}
Z.~{Zhou}, X.~{Chen}, E.~{Li}, L.~{Zeng}, K.~{Luo}, and J.~{Zhang}, ``Edge
  intelligence: Paving the last mile of artificial intelligence with edge
  computing,'' \emph{Proc. IEEE}, vol. 107, no.~8, pp. 1738--1762, 2019.

\bibitem{gxzhu_2018_edge_learning}
G.~{Zhu}, D.~{Liu}, Y.~{Du}, C.~{You}, J.~{Zhang}, and K.~{Huang}, ``Toward an
  intelligent edge: Wireless communication meets machine learning,'' \emph{IEEE
  Commun. Mag.}, vol.~58, no.~1, pp. 19--25, 2020.

\bibitem{niyato2020feel}
W.~Y.~B. {Lim}, N.~C. {Luong}, D.~T. {Hoang}, Y.~{Jiao}, Y.~C. {Liang},
  Q.~{Yang}, D.~{Niyato}, and C.~{Miao}, ``Federated learning in mobile edge
  networks: A comprehensive survey,'' \emph{IEEE Commun. Surveys Tuts.},
  vol.~22, no.~3, pp. 2031--2063, 2020.

\bibitem{wang2019adaptive}
S.~{Wang}, T.~{Tuor}, T.~{Salonidis}, K.~K. {Leung}, C.~{Makaya}, T.~{He}, and
  K.~{Chan}, ``Adaptive federated learning in resource constrained edge
  computing systems,'' \emph{IEEE J. Sel. Areas Commun.}, vol.~37, no.~6, pp.
  1205--1221, 2019.

\bibitem{deniz2019federated_edge_learning}
M.~M. {Amiri} and D.~{Gündüz}, ``Machine learning at the wireless edge:
  Distributed stochastic gradient descent over-the-air,'' \emph{IEEE Trans.
  Signal Process.}, vol.~68, pp. 2155--2169, 2020.

\bibitem{gxzhu2018FEEL}
G.~{Zhu}, Y.~{Wang}, and K.~{Huang}, ``Broadband analog aggregation for
  low-latency federated edge learning,'' \emph{IEEE Trans. Wireless Commun.},
  vol.~19, no.~1, pp. 491--506, 2020.

\bibitem{yang2019aircomp}
K.~{Yang}, T.~{Jiang}, Y.~{Shi}, and Z.~{Ding}, ``Federated learning via
  over-the-air computation,'' \emph{IEEE Trans. Wireless Commun.}, vol.~19,
  no.~3, pp. 2022--2035, 2020.

\bibitem{chen2019joint}
M.~{Chen}, Z.~{Yang}, W.~{Saad}, C.~{Yin}, H.~V. {Poor}, and S.~{Cui}, ``A
  joint learning and communications framework for federated learning over
  wireless networks,'' \emph{IEEE Trans. Wireless Commun.}, vol.~20, no.~1, pp.
  269--283, 2021.

\bibitem{jinke2021}
J.~{Ren}, G.~{Yu}, and G.~{Ding}, ``Accelerating dnn training in wireless
  federated edge learning systems,'' \emph{IEEE J. Sel. Areas Commun.},
  vol.~39, no.~1, pp. 219--232, 2021.

\bibitem{dingzhu2020}
D.~{Wen}, M.~{Bennis}, and K.~{Huang}, ``Joint parameter-and-bandwidth
  allocation for improving the efficiency of partitioned edge learning,''
  \emph{IEEE Trans. Wireless Commun.}, vol.~19, no.~12, pp. 8272--8286, 2020.

\bibitem{yang2020}
H.~H. {Yang}, Z.~{Liu}, T.~Q.~S. {Quek}, and H.~V. {Poor}, ``Scheduling
  policies for federated learning in wireless networks,'' \emph{IEEE Trans.
  Commun.}, vol.~68, no.~1, pp. 317--333, 2020.

\bibitem{yq2020quantization}
Y.~{Du}, S.~{Yang}, and K.~{Huang}, ``High-dimensional stochastic gradient
  quantization for communication-efficient edge learning,'' \emph{IEEE Trans.
  Signal Process.}, vol.~68, pp. 2128--2142, 2020.

\bibitem{sun2020energy}
Y.~{Sun}, S.~{Zhou}, and D.~{Gündüz}, ``Energy-aware analog aggregation for
  federated learning with redundant data,'' in \emph{IEEE Int. Conf. Commun.
  (ICC)}, Dublin, Ireland, Jun 7-11, 2020.

\bibitem{yang2019energy}
Z.~{Yang}, M.~{Chen}, W.~{Saad}, C.~S. {Hong}, and M.~{Shikh-Bahaei}, ``Energy
  efficient federated learning over wireless communication networks,''
  \emph{{to appear in} IEEE Trans. Wireless Commun.}, 2020.

\bibitem{mo2020energyefficient}
X.~Mo and J.~Xu, ``Energy-efficient federated edge learning with joint
  communication and computation design,'' \emph{[Online]
  https://arxiv.org/pdf/2003.00199.pdf}, 2020.

\bibitem{qs2020cpu-gpu}
Q.~Zeng, Y.~Du, K.~Huang, and K.~K. Leung, ``Energy-efficient resource
  management for federated edge learning with cpu-gpu heterogeneous
  computing,'' \emph{[Online] https://arxiv.org/pdf/2007.07122.pdf}, 2020.

\bibitem{clerckx2019fundamental}
B.~{Clerckx}, R.~{Zhang}, R.~{Schober}, D.~W.~K. {Ng}, D.~I. {Kim}, and H.~V.
  {Poor}, ``Fundamentals of wireless information and power transfer: From {RF}
  energy harvester models to signal and system designs,'' \emph{IEEE J. Sel.
  Areas Commun.}, vol.~37, no.~1, pp. 4--33, 2019.

\bibitem{clerckx2021wireless}
B.~Clerckx, K.~Huang, L.~R. Varshney, S.~Ulukus, and M.-S. Alouini, ``Wireless
  power transfer for future networks: Signal processing, machine learning,
  computing, and sensing,'' \emph{[Online]
  https://arxiv.org/pdf/2101.04810.pdf}, 2021.

\bibitem{kaibin2015magazine}
K.~{Huang} and X.~{Zhou}, ``Cutting the last wires for mobile communications by
  microwave power transfer,'' \emph{IEEE Commun. Mag.}, vol.~53, no.~6, pp.
  86--93, 2015.

\bibitem{zhang2013mimo}
R.~{Zhang} and C.~K. {Ho}, ``Mimo broadcasting for simultaneous wireless
  information and power transfer,'' \emph{IEEE Trans. Wireless Commun.},
  vol.~12, no.~5, pp. 1989--2001, 2013.

\bibitem{ju2014WPT}
H.~{Ju} and R.~{Zhang}, ``Throughput maximization in wireless powered
  communication networks,'' \emph{IEEE Trans. Wireless Commun.}, vol.~13,
  no.~1, pp. 418--428, 2014.

\bibitem{kaibin2014powerbeacon}
K.~{Huang} and V.~K.~N. {Lau}, ``Enabling wireless power transfer in cellular
  networks: Architecture, modeling and deployment,'' \emph{IEEE Trans. Wireless
  Commun.}, vol.~13, no.~2, pp. 902--912, 2014.

\bibitem{energy-efficiency}
Q.~{Wu}, M.~{Tao}, D.~W. {Kwan Ng}, W.~{Chen}, and R.~{Schober},
  ``Energy-efficient resource allocation for wireless powered communication
  networks,'' \emph{IEEE Trans. Wireless Commun.}, vol.~15, no.~3, pp.
  2312--2327, 2016.

\bibitem{rate-energy-tradeoff}
X.~{Zhou}, R.~{Zhang}, and C.~K. {Ho}, ``Wireless information and power
  transfer: Architecture design and rate-energy tradeoff,'' \emph{IEEE Trans.
  Commun.}, vol.~61, no.~11, pp. 4754--4767, 2013.

\bibitem{zhu2020onebit}
G.~{Zhu}, Y.~{Du}, D.~{Gündüz}, and K.~{Huang}, ``One-bit over-the-air
  aggregation for communication-efficient federated edge learning: Design and
  convergence analysis,'' \emph{{to appear in} IEEE Trans. Wireless Commun.},
  2020.

\bibitem{yu2019parallel}
H.~{Yu}, S.~{Yang}, and S.~{Zhu}, ``Parallel restarted {SGD} with faster
  convergence and less communication: Demystifying why model averaging works
  for deep learning,'' in \emph{Proc. AAAI Conf. Artif. Intell.}, Honolulu,
  USA, Jan 27 - Feb 1, 2019.

\bibitem{basu2019qsparse}
D.~Basu, D.~Data, C.~Karakus, and S.~Diggavi, ``Qsparse-local-{SGD}:
  Distributed {SGD} with quantization, sparsification and local computations,''
  in \emph{Proc. Adv. Neural Inf. Process. Syst. (NeurIPS)}, Vancouver, Canada,
  Dec 8-14, 2019.

\bibitem{koloskova2019decentralized}
A.~Koloskova, S.~Stich, and M.~Jaggi, ``Decentralized stochastic optimization
  and gossip algorithms with compressed communication,'' in \emph{Proc. Int.
  Mach. Learn. Res. (ICLR)}, Long Beach, USA, Jun 9-15, 2019.

\bibitem{ijcai2018}
F.~Zhou and G.~Cong, ``On the convergence properties of a k-step averaging
  stochastic gradient descent algorithm for nonconvex optimization,'' in
  \emph{Proc. Int. Joint Conf. Artif. Intell., ({IJCAI})}, Stockholm, Sweden,
  Jul 13-19, 2018.

\bibitem{baccelli2006MPTmodel}
F.~{Baccelli}, B.~{Blaszczyszyn}, and P.~{Muhlethaler}, ``An aloha protocol for
  multihop mobile wireless networks,'' \emph{IEEE Trans. Inf. Theory}, vol.~52,
  no.~2, pp. 421--436, 2006.

\bibitem{mcmahan2017federatedlearning}
B.~McMahan, E.~Moore, D.~Ramage, S.~Hampson, and B.~A. y~Arcas,
  ``Communication-efficient learning of deep networks from decentralized
  data,'' in \emph{Proc. Int. Conf. Artif. Intell. Statist. (AISTATS)}, Fort
  Lauderdale, USA, Apr 20-22, 2017.

\bibitem{zhang2018flops}
X.~{Zhang}, X.~{Zhou}, M.~{Lin}, and J.~{Sun}, ``Shufflenet: An extremely
  efficient convolutional neural network for mobile devices,'' in \emph{Proc.
  IEEE/CVF Conf. Comput. Vision Pattern Recognit. (CVPR)}, Salt Lake City, USA,
  Jun 18-23, 2018.

\bibitem{liu2012dvfs}
{C. Liu}, {J. Li}, {W. Huang}, J.~{Rubio}, E.~{Speight}, and {F. Lin},
  ``Power-efficient time-sensitive mapping in heterogeneous systems,'' in
  \emph{Proc. Int. Conf. Parallel Archit. Compilation Tech. (PACT)},
  Minneapolis, USA, Sep 21-25, 2012.

\bibitem{bernstein2018signsgd}
J.~Bernstein, Y.-X. Wang, K.~Azizzadenesheli, and A.~Anandkumar, ``sign{SGD}:
  Compressed optimisation for non-convex problems,'' in \emph{Proc. Int. Conf.
  Mach. Learn.}, Stockholm, Sweden, Jul 10-15, 2018.

\bibitem{sattler2020noniid}
F.~{Sattler}, S.~{Wiedemann}, K.~R. {Müller}, and W.~{Samek}, ``Robust and
  communication-efficient federated learning from non-i.i.d. data,'' \emph{IEEE
  Trans. Neural Netw. Learn. Syst.}, vol.~31, no.~9, pp. 3400--3413, 2020.

\bibitem{zhu2018nonconvex}
Z.~Allen-Zhu, ``Natasha 2: Faster non-convex optimization than {SGD},'' in
  \emph{Proc. Adv. Neural Inf. Process. Syst. (NeurIPS)}, Montreal, Canada, Dec
  2-8, 2018.

\bibitem{bottou2018}
L.~Bottou, F.~E. Curtis, and J.~Nocedal, ``Optimization methods for large-scale
  machine learning,'' \emph{SIAM Rev.}, vol.~60, no.~2, pp. 223--311, 2018.

\bibitem{wu2020noisySGD}
J.~Wu, W.~Hu, H.~Xiong, J.~Huan, V.~Braverman, and Z.~Zhu, ``On the noisy
  gradient descent that generalizes as {SGD},'' \emph{[Online]
  https://arxiv.org/pdf/1906.07405.pdf}, 2019.

\bibitem{Stephan1945statistics}
F.~F. Stephan, ``The expected value and variance of the reciprocal and other
  negative powers of a positive bernoullian variate,'' \emph{Ann. Math.
  Statist.}, vol.~16, no.~1, pp. 50--61, 1945.

\end{thebibliography}
\bibliographystyle{IEEEtran}

\end{document}